\pdfoutput=1
\RequirePackage[l2tabu,orthodox]{nag}
\documentclass[USenglish]{lipics-v2018}\hideLIPIcs\nolinenumbers

\usepackage{microtype}
\usepackage[T1]{fontenc}
\usepackage[utf8]{inputenc}

\usepackage[noend]{algorithmic}
\usepackage{xspace}
\usepackage{xcolor}
\usepackage{tikz}
\usetikzlibrary{shapes,arrows,positioning,decorations.markings,shapes.misc,calc,angles,intersections,quotes,decorations.pathmorphing}

\newcommand{\eps}{\varepsilon}

\newtheorem{fact}[theorem]{Fact}
\newtheorem{conj}[theorem]{Conjecture}

\newcommand{\N}{\mathbf N}
\newcommand{\Z}{\mathbf Z}

\newcommand{\poly}{\operatorname{poly}\xspace}
\newcommand{\polylog}{\operatorname{polylog}\xspace}

\newtheorem*{rep@theorem}{\rep@title}
\newcommand{\newreptheorem}[2]{%
	\newenvironment{rep#1}[1]{%
		\def\rep@title{#2 \ref{##1}}%
		\begin{rep@theorem}[restated]}%
		{\end{rep@theorem}}}

\newreptheorem{theorem}{Theorem}
\newreptheorem{lemma}{Lemma}
\newreptheorem{proposition}{Proposition}
\newreptheorem{corollary}{Corollary}

\newcommand{\gate}[1]{\ensuremath{\mathrm{#1}}\xspace}

\newcommand{\TC}[1]{\ensuremath{\mathrm{TC}^{#1}}\xspace}
\newcommand{\AC}[1]{\ensuremath{\mathrm{AC}^{#1}}\xspace}

\usepackage{mathtools}
\DeclarePairedDelimiter\paren{\lparen}{\rparen}
\DeclarePairedDelimiter\abs{\lvert}{\rvert}
\DeclarePairedDelimiter\set{\{}{\}}
\DeclarePairedDelimiterX\setc[2]{\{}{\}}{#1 \colon #2}
\DeclarePairedDelimiterX\parenc[2]{\lparen}{\rparen}{\,#1 \;\delimsize\vert\; #2\,}

\usepackage{enumitem}
\setlist[itemize]{label=$\circ$}

\newenvironment{algor}[3]{%
\bigskip
\noindent{\sffamily\bfseries Algorithm #1} ({\itshape#2\/}) {\itshape #3}
\begin{description}[labelindent=0pt,labelwidth=1.2em,labelsep=4pt,leftmargin=!]%
\vspace{-.8ex}}{%
\end{description}\medskip}

\author
{Amir Abboud}
{IBM Almaden Research Center, San Jose, CA, USA}
{amir.abboud@ibm.com}
{}{}

\author
{Karl Bringmann}
{Max Planck Institute for Informatics, Saarland Informatics Campus, Saarbrücken, Germany}
{kbringma@mpi-inf.mpg.de}
{}{}

\author
{Holger Dell}
{Saarland University and Cluster of Excellence (MMCI), Saarbrücken, Germany}
{hdell@mmci.uni-saarland.de}
{https://orcid.org/0000-0001-8955-0786}
{}

\author
{Jesper Nederlof}
{Eindhoven University of Technology, Eindhoven, The Netherlands}
{j.nederlof@tue.nl}
{}{}

\authorrunning{Amir Abboud, Karl Bringmann, Holger Dell, and Jesper Nederlof}
\title{More Consequences of Falsifying SETH\newline and the Orthogonal Vectors Conjecture}
\titlerunning{More Consequences of Falsifying SETH and the Orthogonal Vectors Conjecture}

\relatedversion{Proceedings version
  \href
  {http://dx.doi.org/10.1145/3188745.3188938}
  {doi:10.1145/3188745.3188938}
  to appear at the 50th Annual ACM SIGACT Symposium on the Theory of Computing, June 25--29, 2018, Los Angeles, CA, USA.
}
\subjclass{
}
\category{}

\begin{document}

	\maketitle

  \begin{abstract}

The Strong Exponential Time Hypothesis and the OV-conjecture are two popular
hardness assumptions used to prove a plethora of lower bounds, especially in the
realm of polynomial-time algorithms.
The OV-conjecture in moderate dimension states there is no $\eps>0$ for which an~$O(N^{2-\eps})\;\poly(D)$ time algorithm can decide whether there is a pair of
orthogonal vectors in a given set of size~$N$ that contains $D$-dimensional
binary vectors.

We strengthen the evidence for these hardness assumptions. In particular, we
show that if the OV-conjecture fails, then two problems for which we are far from
obtaining even tiny improvements over exhaustive search would have surprisingly
fast algorithms. If the OV conjecture is false, then there is a fixed~$\eps>0$ such
that:
\begin{enumerate}
  \item For all~$d$ and all large enough~$k$, there is a randomized algorithm that takes
  $O(n^{(1-\eps)k})$ time to solve the Zero-Weight-$k$-Clique
  and Min-Weight-$k$-Clique problems on $d$-hypergraphs with~$n$ vertices.
  As a consequence, the OV-conjecture is implied by the Weighted Clique conjecture.
  \item
    For all $c$, the satisfiability of sparse \TC1 circuits on $n$ inputs
    (that is, circuits with $cn$ wires, depth $c\log n$, and negation, AND, OR, and threshold gates) can be computed in time~${O((2-\eps)^n)}$.
\end{enumerate}

  \end{abstract}
	

\section{Introduction}
The Strong Exponential Time Hypothesis (SETH) is a cornerstone of contemporary algorithm design that was formulated by Impagliazzo and Paturi~\cite{DBLP:journals/jcss/ImpagliazzoP01} and recently gained extensive popularity. 
It postulates that exhaustive search is essentially the fastest possible method to decide the satisfiability of bounded-width CNF formulas.
SETH is used in the study of exact and fixed parameter tractable algorithms, see
e.g~\cite{DBLP:journals/talg/CyganDLMNOPSW16, DBLP:conf/soda/PatrascuW10} or the
book by Cygan et al.~\cite{DBLP:books/sp/CyganFKLMPPS15}.
In this area, it implies, among other things, tight lower bounds for problems on graphs that have small treewidth or pathwidth~\cite{DBLP:conf/soda/LokshtanovMS11a,DBLP:conf/focs/CyganNPPRW11, DBLP:conf/stoc/CyganKN13}.

Closely related to SETH, 
the orthogonal vectors problem (OV) is, given two sets~$A$ and~$B$ of~$N$ vectors from
$\set{0,1}^D$, to decide whether there are vectors~$a\in A$ and $b\in B$ such
that $a$ and $b$ are orthogonal in~$\Z^D$.
If $D\le O(N^{0.3})$ holds, the problem can be solved in time $\tilde{O}(N^2)$ using an algorithm based on fast rectangular matrix multiplication (see e.g.~\cite{DBLP:conf/focs/Gall12}).
SETH implies~\cite{Wil05} that this algorithm is essentially as fast as possible; in particular, SETH implies the following hardness conjecture, which was given its name by Gao et al.~\cite{GIKW17}.

\begin{conj}[Moderate-dimension OV Conjecture]
  \label{conj OVC}
  \hspace{-1mm}
  There are no reals $\eps,\delta>0$ such that~OV for $D=N^{\delta}$ can be
  solved\footnote{In this work we hardly distinguish between randomized and deterministic algorithms, as even randomized algorithms with the desired running times would constitute an important breakthrough.}
  in time $O(N^{2-\eps})$.
\end{conj}
The moderate-dimension OV conjecture is used to study the fine-grained complexity of problems in~P, for which it has remarkably strong and diverse implications.
If the conjecture is true, then dozens of important problems from all across computer science exhibit running time lower bounds that match existing upper bounds up to subpolynomial factors.
These include pattern matching and other problems in bioinformatics~\cite{AVW14,BI15,KPS17,ABBK17}, graph algorithms~\cite{RV13,AV14,GIKW17}, computational geometry~\cite{Bring14}, formal languages \cite{BI16,BGL16}, time-series analysis \cite{ABV15a,BK15}, and even economics~\cite{MPS16} (see~\cite{Vass15} for a more comprehensive list).

Gao et al.~\cite{GIKW17} also named the \emph{low-dimension OV conjecture}, which asserts that OV does not have subquadratic algorithms whenever $D=\omega(\log N)$ holds.
The low-dimension implies the moderate-dimension variant of the OV conjecture, and both are implied by~SETH~\cite{Wil05}.
Recent results on the hardness of approximation problems, such as Maximum Inner Product~\cite{ARW17}, rely on the stronger conjecture (perhaps also \cite{BIL17,BRSV17a}).
However, for the vast majority of OV-based hardness results, reducing the dimension only affects lower-order terms in the lower bounds and so it often suffices to assume the moderate-dimension variant.
Doing so makes results stronger, and it is this variant of the OV conjecture that we strengthen further in the present work.

\tikzstyle{vecArrow} = [thick,
decoration={markings,mark=at position 1 with {\arrow[semithick]{open triangle 60}}},
double distance=1.4pt, shorten >= 5.5pt,
preaction = {decorate},
postaction = {draw,line width=1.4pt, white,shorten >= 4.5pt}]%
\tikzstyle{innerWhite} = [semithick, white,line width=1.4pt, shorten >= 4.5pt]%
\begin{figure*}[pt]
  \centering
  \resizebox{\columnwidth}{!}{%
    \begin{tikzpicture}[thick, every node/.style={scale=0.75},node distance=2cm,inner sep = 7pt,every node/.style=on grid]
      \begin{scope}
        \node[draw,rounded rectangle,align=center] (a) {Min-Weight\\$k$-Clique};
        \node[draw,rounded rectangle,below of=a,align=center] (b) {Min-Weight\\3-Clique};
        \node[draw,rounded rectangle,right of=b,align=center,xshift=.7cm] (c) {APSP};
        \node[draw,rounded rectangle,right of=c,align=center,xshift=1cm] (d) {OV in dimension $N^\delta$};
        \node[draw,rounded rectangle,above of=d,align=center] (e) {CNF-SAT};
        \node[cloud,cloud puffs=20,cloud puff arc=110,aspect=2,inner sep=0mm,draw,below of=b,align=center,yshift=-0.3cm] (f) {Various Problems\\on Graphs};
        \node[cloud,cloud puffs=20,cloud puff arc=110,aspect=2,inner sep=-1mm,draw,below of=d,align=center,yshift=-0.3cm] (g) {Various Problems on\\Strings, Graphs, etc.};
        \node[cloud,cloud puffs=20,cloud puff arc=110,aspect=2,inner sep=-4mm,draw,left of=b,align=center,xshift=-3cm] (h) {Maximum Weight Rectangle \\ Improving Viterbi's Algorithm\\ Tree Edit Distance};

        \draw[->] (a) to (b);
        \draw[<->] (b) to (c);
        \draw[->] (e) to (d);
        \draw[->] (d) to (g);
        \draw[->] (b) to (f);
        \draw[->] (a) to (h);
        \draw[->,line width=1mm,shorten <=0.2cm,shorten >=0.3cm] (a) -- (d) node[midway,sloped,above] {Theorem~\ref{theorem: OV main}};
      \end{scope}
    \end{tikzpicture}
  }
	\caption{\label{fig1}%
  Illustration of the landscape of Hardness in P.
  An arrow from problem A to problem B indicates that improving the runtime of problem B from $T_B$ to $T_B^{1-\eps}$ implies an improvement for problem A from $T_A$ to $T_A^{1-\eps'}$.
  Our contribution is the bold black arrow (Theorem~\ref{theorem: OV main}).
  }
\end{figure*}
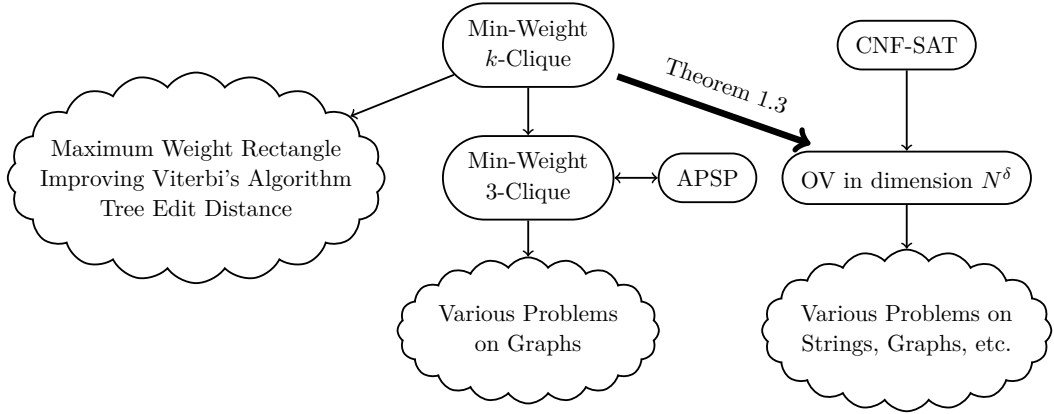
\subsection{Other conjectures}
Two other popular conjectures in fine-grained complexity make assertions for the All-Pairs-Shortest-Path Problem~(APSP)%
\footnote{The APSP problem is to compute all pairwise distances in a graph (given by its adjacency matrix) on $n$ nodes and with edges weights in some polynomial range. It is conjectured to require $n^{3-o(1)}$ time, and many problems, especially on graphs, are known to be equivalent to or at least as hard as APSP, see e.g.\ \cite{RZ04,VW09,AV14,AGV15,Saha15,AVY15,Dahlgaard16}.}
and the $3$-SUM problem%
\footnote{The $3$-SUM problem is to decide if a given set of $n$ integers contains three that sum to zero. It is conjectured that the problem requires $n^{2-o(1)}$ time. Many problems, especially in computational geometry, are known to be $3$-SUM-hard, see \cite{DBLP:journals/comgeo/GajentaanO95}.}.
It is an important and longstanding open question to determine the relationship between APSP, $3$-SUM, and OV.
In particular, it is open whether the APSP conjecture or the $3$-SUM conjecture imply the moderate-dimension OV conjecture.

Closely related to APSP is the Min-Weight-$k$-Clique problem:
Given a graph on $n$ nodes with integer edge-weights in some polynomial range, the goal is to find a $k$-clique of minimum weight.
The exhaustive search algorithm solves this problem in $O(n^k)$ time,
and for $k=3$, the problem is subcubically-equivalent to the APSP conjecture~\cite{DBLP:conf/focs/WilliamsW10}.
That is, either both APSP and Min-Weight-$3$-Clique have algorithms running in time $O(n^{3-\eps})$ for some $\eps > 0$, or neither of them do.

For all integers $k \geq 3$, there is a simple reduction from Min-Weight-$k$-Clique to Min-Weight-$3$-Clique, and by combining it with the fastest known APSP algorithm~\cite{Wil14,CW16}, we can solve Min-Weight-$k$-Clique in time $n^{k}/\exp\paren[\big]{\Omega(\sqrt{\log{n}})}$.
The improvement of this algorithm over exhaustive search is subpolynomial.
Due to the equivalence with APSP for $k=3$, it is natural to conjecture that a truly polynomial advantage in the running time is impossible.

\begin{conj}[Weighted Clique Conjecture]\label{conj weighted clique}
  There is no real $\eps>0$ and integer $k\ge 3$ such that the Min-Weight-$k$-Clique problem on $n$-vertex graphs and with edge-weights in $\{-M,\ldots,M\}$ can be solved in time $O(n^{(1-\eps)k}) \polylog{M}$.%
\end{conj}

This conjecture implies the APSP conjecture, and thus implies lower bounds for all problems that are known to be APSP-hard~\cite{RZ04,VW09,AV14,AGV15,Saha15,AVY15,Dahlgaard16}.
In addition, the Weighted Clique conjecture implies lower bounds for a variety of problems that are not known to be APSP-hard:
the Local Alignment problem
\cite{AVW14}
from bioinformatics, 
the Maximum Rectangle problem
\cite{BDT16}
from computational geometry,
the Viterbi problem
\cite{BT16}
from machine learning and,
the Tree Edit Distance problem
\cite{BGMW17}.

\subsection{Our Results for OV}

We prove that the Weighted Clique conjecture (Conjecture~\ref{conj weighted clique}) implies the moderate-dimension OV conjecture (Conjecture~\ref{conj OVC}).
To this end, we design a tight randomized reduction from Min-Weight-$k$-Clique to OV.
The impact of this result on fine-grained complexity in P is depicted in Figure~\ref{fig1}.
As can be seen, we identify Min-Weight-$k$-Clique as a core problem in this landscape, since it tightly reduces to most problems that have known conditional lower bounds; the main exceptions are 3-SUM-hard problems and the few problems that do require the low-dimension version of the OV conjecture. 

In fact, our result is even stronger: We show that improved algorithms for moderate-dimension OV leads to improved algorithms for finding weighted cliques \emph{even in hypergraphs}.
A $d$-hypergraph is a hypergraph in which all edges are of size at most $d$.
A \emph{clique} of a $d$-hypergraph $G$ is a subset $X \subseteq V(G)$ such that for every $e \subseteq X$ of size at most $d$ we have $e \in E(G)$.
The \mbox{\emph{$k$-Clique}} problem is given~$G$ as input to find such a clique~$X$ of size~$k$.
We also study weighted versions, where we are additionally given an edge-weight function $w:E(G)\to\Z$ and a target integer~${t\in\Z}$.
The weight of a clique~$X$ in~$G$ is the sum $\sum_{e} w(e)$ over all edges~$e\in E(G)$ with $e\subseteq X$.
In the \emph{Exact-Weight-$k$-Clique} problem, we need to find a $k$-clique~$X$ of weight exactly~$t$, and in \emph{Min-Weight-$k$-Clique} we need one of weight at most~$t$.
We are ready to formally state our first theorem.
\def\theoremOVmain{%
  If the moderate-dimension OV conjecture is false, then
  there exists an $\eps>0$ such that for every integer~$d$ there is a (large) integer $k=k(d,\eps)$ satisfying the following statements:
	\begin{itemize}
		\item $k$-Clique can be solved on $d$-hypergraphs in time $O(n^{(1-\eps)k})$.
		\item Exact-Weight-$k$-Clique and Min-Weight-$k$-Clique can be solved on $d$-hypergraphs with weights in
      $\set{-M,\dots,M} $ in randomized time~$O(n^{(1-\eps)k})\cdot\polylog M$.
	\end{itemize}%
}
\begin{theorem}\label{theorem: OV main}
  \theoremOVmain
\end{theorem}

Clique problems on hypergraphs appear to be harder than on graphs.
For example, in the unweighted case, $k$-Clique on graphs can be solved in $O(n^{0.79 k})$ time \cite{NP85,EG04}, whereas on $3$-hypergraphs no polynomial improvement over the $O(n^k)$ exhaustive search algorithm is known.
If the moderate-dimension OV conjecture is false, then Theorem~\ref{theorem: OV main} implies that $k$-Clique on $3$-hy\-per\-graphs does have such an improvement for some large enough integer~$k$.
This strengthens the OV conjecture.

Significantly improved algorithms for $k$-Clique on $d$-hypergraphs are not only unknown, they are in fact known to imply breakthrough algorithms for Max-$d$-SAT, the optimization version of the $d$-CNF-SAT problem that needs to find an assignment that satisfies as many clauses as possible.
Using this known reduction, we obtain the following corollary to Theorem~\ref{theorem: OV main}.

\def\corollarymaxsat{%
  If the moderate-dimension OV conjecture is false, then there exists an $\eps>0$ such that, for all integers~$d$, there is an algorithm for Max-$d$-SAT that runs in time $O^*\paren[\big]{(2-\eps)^n}$.%
}
\begin{corollary}\label{cor: maxsat}
  \corollarymaxsat
\end{corollary}
Due to this corollary, the moderate-dimension OV conjecture is not only implied by SETH, but even by its stronger cousin for the Max-SAT problem on bounded-width CNF.
As a consequence, all known lower bounds based on the moderate-dimension OV conjecture are now automatically based on the hardness of bounded-width Max-SAT rather than just bounded-width CNF-SAT.
Corollary~\ref{cor: maxsat} subsumes some results~\cite{ABV15a,AVY15,KT17} where this was done in special cases.

The implications of Theorem~\ref{theorem: OV main} for the weighted problems strengthen the moderate-dimension OV conjecture further.
Any algorithm that solves Exact-Weight-$k$-Clique can in particular solve Zero-Weight-$k$-Clique where the target weight satisfies $t=0$.
In turn, any algorithm for the latter problem can also be used to solve Min-Weight-$k$-Clique without significant running time overhead~\cite{DBLP:conf/mfcs/NederlofLZ12}.
While the best known algorithms for Min-Weight-$k$-Clique run in time
$n^{k}/\exp\paren[\big]{\Omega(\sqrt{\log{n}})}$ \cite{Wil14,CW16},  such superpolylogarithmic shavings are open for Zero-Weight-$k$-Clique.
The $k=3$ case is particularly interesting, since solving Zero-Weight-$3$-Clique in $O(n^{3-\eps})$ time refutes not only the APSP conjecture but \emph{also} the $3$-SUM conjecture \cite{VW09,Pat10,KPP14}.

\paragraph*{Proof ideas}
We prove Theorem~\ref{theorem: OV main} by designing a tight reduction from Min-Weight-$k$-Clique on $d$-hy\-per\-graphs to~OV.
We sketch the reduction for~$d=2$.
It has two main \emph{stages}.
In the first, we reduce Min-Weight-$k$-Clique on graphs to unweighted $k$-Clique on 4-hy\-per\-graphs, where each hyperedge has cardinality at most~$4$.
We achieve this in a sequence of weight reduction \emph{steps}:
We start with a standard hashing trick to reduce the weights to a polynomial range. Then, to reduce the weights further, we chop the bits of the numbers into vectors and use a squaring trick to combine all the coordinates.
This trick is borrowed from~\cite{AbboudLW14}, where it was used to reduce node weights in graphs. We instead use it  to reduce edge weights, which however we only achieve by transforming the graph into a $4$-hypergraph.
Finally, once the weights are small enough, we remove them completely via an exhaustive search.

In the second stage, we reduce the unweighted $k$-Clique problem on $4$-hypergraphs to a $k$-wise variant of OV.
The reduction maps each node to a Boolean vector by encoding the incident hyperedges into the coordinates such that a disjointness check among $k$ vectors corresponds to checking that $k$ nodes form a hyperclique.
Finally, we reduce the $k$-wise variant to the OV problem using a standard reduction.

\subsection{Our results for SETH and CNF-SAT}
There are many algorithms that solve $d$-CNF-SAT, the satisfiability problem on $d$-CNF formulas, in time $O^*\paren[\big]{(2-\eps_d)^n}$.
As $d$ grows to infinity, the constants~$\eps_d$ for all these algorithms tend to~$0$ in the limit.
SETH was conceived by this observation, and it asserts exactly this:
If~$\eps_d$ is the largest real such that $d$-CNF-SAT can be solved in time $O^*((2-\eps_d)^n)$, then $\lim_{d\to\infty} \eps_d=0$ holds.
Thus SETH is not about the hardness of an individual problem, but about a sequence of problems each of which we know to have a faster algorithm than exhaustive search.
This makes it easier prove lower bounds under SETH, but it is unfortunate if we want to have confidence that SETH and the implied lower bounds are true.

Indeed, there are  algorithms that get substantial $n^{\omega(1)}$ speed-ups
over $2^n$ for CNF formulas of unbounded width (see e.g.~\cite{DBLP:series/faia/2009-185, CIP06}).
If instead of CNF formulas, we consider more complex Boolean circuits, such as bounded-depth threshold circuits (\TC0-circuits), then we should get a computationally harder satisfiability problem.
For linear-size threshold circuits of depth two, there are satisfiability algorithms that run in time $O^*\paren[\big]{(2-\eps_c)^n}$; here too, the constants~$\eps_c$ tend to~$0$ as~$c$ grows~\cite{DBLP:conf/focs/ImpagliazzoPS13,DBLP:conf/sat/ChenS15}.
However, even for linear-size threshold circuits of depth~$4$, satisfiability algorithms with speed-up $n^{\omega(1)}$ are unknown.
Obtaining such algorithms for linear-size \TC0 would resolve Williams' question~\cite{DBLP:journals/jacm/Williams14} of whether his circuit lower bound framework can prove $\mathrm{NEXP}\nsubseteq\TC0$.
We show that a refutation of SETH would constitute progress on these questions.
\def\theoremTCzero{%
  If SETH fails, then there is an $\eps>0$ such that, for all constants $c$ and
  $d$, the satisfiability of depth-$d$ threshold circuits with~$cn$ wires
  can be determined in time $O^*\paren[\big]{(2-\eps)^n}$.%
}
\begin{theorem}\label{theorem: TC0}
  \theoremTCzero
\end{theorem}

This theorem is the newest member in a sequence of increasingly general results of Santhanam and Srinivasan~\cite{San12}, Dantsin and Wolpert~\cite{DBLP:conf/ciac/DantsinW13}, and Cygan et al.~\cite{DBLP:journals/talg/CyganDLMNOPSW16}, who show that refuting SETH implies faster algorithms for the satisfiability of linear-size formulas, linear-size \AC0-circuits, and linear-size VSP-circuits, respectively.
Our class of linear-size \TC0-circuits contains the classes of linear-size formulas and \AC0-circuits, so we generalize these two results.
The class VSP is less understood and little is known about its complexity properties.
While algorithms with running time~$2^n/n^{\omega(1)}$ are known for the satisfiability of linear-size \AC0-circuits~\cite{DBLP:conf/soda/ImpagliazzoMP12}, such algorithms are not known for linear-size \TC0-circuits, even when the depth is $4$ and the number of wires is~$10n$.

As usual with reductions, a pious believer is biased to view Theorem~\ref{theorem: TC0} as another confirmation that SETH and all its logical implications are indeed true, which includes the moderate-dimension OV conjecture.
On the other hand, a skeptic who is hesitant to believe SETH or one of its implications is now invited to start their refutation attempt by providing faster algorithms for linear-size \TC0-circuits, since any refutation of SETH would have to do that implicitly.

\subsubsection*{Extension for CNF-SAT}
Much like most OV-based lower bounds in P can be based on the moderate-dimension OV conjecture rather than its low-dimension variant, many SETH-based lower bounds for exponential time and parameterized problems can be based on the weaker assumption that satisfiability cannot be solved in time~$O^*\paren[\big]{(2-\eps)^n}$ for CNF formulas of unbounded width.
This weaker assumption for CNF-SAT suffices, for example, in results for graph problems that have small treewidth or pathwidth~\cite{DBLP:conf/soda/LokshtanovMS11a,DBLP:conf/focs/CyganNPPRW11,DBLP:conf/stoc/CyganKN13}.
We add further weight to these hardness results by showing that sufficiently fast
algorithms for CNF-SAT imply improved satisfiability algorithms for linear-size threshold circuits of super-logarithmic depth, which is a larger class than \TC1-circuits.

\def\theoremTCone{%
	If CNF-SAT can be solved in $O^*(2^{(1-\eps)n})$ time for some~$\eps>0$, then
	there is an~$\eps'>0$ such that, for all~$c>0$, there is a $\delta > 0$ such that the
	satisfiability for threshold circuits of depth $(\log n)^{1+\delta}$ and at most~$cn$
	wires can be determined in time $O(2^{(1-\eps')n})$.
}
\begin{theorem}\label{theorem: TC1}
  \theoremTCone
\end{theorem}

If SETH is false, then not every problem in $\mathrm{E^{NP}}$ can be computed by linear-size VSP-circuits~\cite{Wil13,JMV15}.
By Theorem~\ref{theorem: TC1}, solving CNF-SAT in time $O^*\paren[\big]{(2-\eps)^n}$ implies the perhaps more natural result that not every problem in $\mathrm{E^{NP}}$ has linear-size \TC1-circuits.

\paragraph*{Proof idea}
Similar to the analogous result by Cygan et al.~\cite{DBLP:journals/talg/CyganDLMNOPSW16} for VSP-circuits, we use a depth reduction technique introduced by Valiant~\cite{DBLP:conf/mfcs/Valiant77}, which shows that VSP-circuits embed nicely into CNF-formulas.
We use an additional trick that allows us to get rid of threshold gates.


\section{Preliminaries}\label{sec:prel}

\paragraph*{Notation}
The $O^*(\cdot)$ and $\tilde O(\cdot)$ notations omit factors that are polynomial and polylogarithmic in the input size, respectively.
We write~$\Z$ for the integers and~$\N$ for the non-negative integers.
We let $[n]=\{1,\ldots,n\}$ for $n\in\N$.
If~$S$ is a set, we write~$\binom{S}{d}$ for the set of all subsets of~$S$ that
have size exactly~$d$, and $\binom{S}{\le d}$ for the set of all subsets of
size at most~$d$.
A $d$-hypergraph~$G$ for $d\in\N$ is a tuple $(V(G),E(G))$, where
$V(G)$ is a finite set of \emph{vertices} and
$E(G)\subseteq\binom{V(G)}{\le d}$ is a set of
\emph{edges}.
If $G$ is a $d$-hypergraph and $X \subseteq V(G)$, then $G[X]$ denotes the
\emph{subgraph induced by~$X$}, that is, $V(G[X])=X$ and
$E(G[X])=E(G)\cap\binom{X}{\le d}$.
A~set $S\subseteq V(G)$ is called a clique in~$G$ if $E(G[S]) = \binom{S}{\le d}$.
A $k$-clique
is a clique of size $k$.

A \emph{graph} is a $2$-hypergraph.
In contrast to usual graph notation,%
\footnote{We remark that there is an alternative definition of hypergraphs and cliques, where each edge of a $d$-hypergraph has size exactly $d$ instead of at most $d$, and a clique is a set $S$ such that $E(G[S]) = \binom{S}{d}$. The two variants are equivalent in terms of the algorithmic problem of deciding whether $G$ contains a $k$-clique. Indeed, if we want to detect a set $S$ with $E(G[S]) = \binom{S}{d}$, then we can add all sets of size at most $d-1$ to~$E(G)$ and then detect a set $S$ with $E(G[S]) = \binom{S}{\le d}$. Similarly, if we want to detect a set $S$ with $E(G[S]) = \binom{S}{\le d}$, then we can build a new hypergraph $G' = (V(G),E')$ where $E'$ contains all sets $e \subseteq V(G)$ of size $d$ whose every subset is in~$E(G)$, and then detect a set $S$ with $E(G'[S]) = \binom{S}{d}$. Similar equivalences hold for the weighted variants of the $k$-Clique problem. Thus, our choice of a variant is only for notational convenience.}
there are also edges~$\set{v}$ of size 1 and the edge $\emptyset$ of size 0; a $k$-clique is a set $\{v_1,\ldots,v_k\}$ for which every pair $\{v_i,v_j\}$ is in~$E(G)$, every singleton~$\{v_i\}$ is in $E(G)$, and $\emptyset \in E(G)$. This does not significantly change the problem of detecting whether $G$ contains a $k$-clique, since testing whether $\emptyset \in E(G)$ is in constant time, and we can assume without loss of generality that $\{v\} \in E(G)$ for all $v \in V(G)$, by deleting all other vertices.

\paragraph*{CNF-SAT}
The $d$-SAT problem is to determine whether a given $d$-CNF formula has a
satisfying assignment. We denote the number of variables by~$n$ and define~$s_d$ as the real number
\[
  \inf
  \setc[\big]
  {
  \delta>0
  }{
  \text{there is an } O\paren[\big]{2^{\delta n}} \text{ time algorithm for $d$-SAT}
  }
\,.
\]

Let $s_\infty = \lim_{d\rightarrow \infty}s_d$.
Impagliazzo and Paturi's \emph{Strong Exponential Time Hypothesis (SETH)}
postulates that $s_\infty=1$ holds~\cite{DBLP:journals/jcss/ImpagliazzoP01}.

\paragraph*{DAGs and Circuits} 
If $G$ is a directed acyclic graph (DAG), we let~$N^-_G(v)$ denote the set of
in-neighbors of $v$ and let $d^{-}_G(v)$ denote the \emph{in-degree} with $d^{-}_G(v)=|N^-_G(v)|$.
The \emph{depth} of $G$ is the length of the longest directed path in it.

A \emph{Boolean function} is any function~$f:\set{0,1}^d\to\set{0,1}$.
It is \emph{symmetric} if $f(x)=f(y)$ holds for all $x,y\in\set{0,1}^d$ whose
Hamming weight is the same.
Let $B$ be a set of symmetric Boolean functions.
A \emph{(Boolean) circuit~$C$ over a basis~$B$} is a pair $(G,\lambda)$ where~$G$ is a directed acyclic graph and $\lambda \in B^V$ is a labeling of its vertex set~$V$ with elements from~$B$.
We say that~$v$ is a \emph{$\lambda_v$-gate}, and we require that the
in-degree of~$v$ is equal to the arity of $\lambda_v$, that is, we have
$\lambda_v:\set{0,1}^{d^-_G(v)}\to\set{0,1}$.
The edges of $G$ are called \emph{wires}, the in-degree of a gate is called
its \emph{fan-in}, and we write $V(C)$ for~$V(G)$.
The set of \emph{input gates} $I(C)$ or $I(G)$ of~$C$ consists of the vertices
with in-degree~$0$,
and the set of \emph{output gates} $O(C)$ or $O(G)$ of~$C$ consists of the
vertices with out-degree~$0$.
If $x\in\set{0,1}^{I(G)}$ is a setting for the input gates, we define~$C_v(x)$ as the \emph{value of~$C$ at $v \in V$ on input~$x$} inductively: 
If $v\in I(C)$, let $C_v(x)=x_v$, and otherwise,
let $C_v(x)=\lambda_v(C_{v_1}(x),\ldots,C_{v_\ell}(x))$, where
$v_1,\ldots,v_\ell$ denotes the in-neighbors of $v$ in $G$; note that this is
well-defined since $G$ is acyclic and $\lambda_v$ is symmetric.
Slightly abusing notation, we may write~$C$ also for the
function~$C:\set{0,1}^{I(G)}\to\set{0,1}^{O(G)}$ with $C(x)=(C_{v})_{v\in
O(G)}$. Or we may view circuits as mapping integers to integers in a fixed
range $[r]$ for convenience while in fact this is implemented by storing the
binary representation of these values with $\lceil \lg r \rceil$ gates.

A $(u,v)$-path in $C$ is a directed path~$u_1,\dots,u_\ell$ in~$G$ with $u=u_1$ and $u_\ell=v$.
If $A \subseteq V(C)$, we let $R_{C}(A,v)$ denote the set of vertices
from which $v$ is reachable without using vertices of $A$, that
is,
\begin{equation} \label{eq:defRCAv}
  R_{C}(A,v)
  =
  \setc
  { u \in V }
{
  \text{
    $G[V\setminus A \cup \set{u,v}]$
    contains $(u,v)$-path
    }}
.
\end{equation}
Finally, for a circuit~$C$, a gate~$v\in V(C)$, and a set~$A\subseteq V(C)$, we
define $C_{v,A}$ as the subcircuit of~$C$ that is induced by the set
$R_C(A,v)$; note that~$v$ is the only output gate of~$C_{v,A}$ and its input gates are contained in~$A\cup I(C)$.

We use the Boolean functions $\gate{NEG}(x) = \neg x$, $\gate{AND}(x,y)=x \wedge
y$, $\gate{OR}(x,y)=x\vee y$ and $\gate{TH}_\theta: \{0,1\}^d \rightarrow
\{0,1\}$ which is, for every positive $\theta \leq n$ defined to be $1$ if
$\sum_{i=1}^d x_i \geq \theta$ and to be $0$ otherwise.
Note that $\gate{AND}(x,y)=\gate{TH}_{2}(x,y)$ and $\gate{OR}(x,y)=\gate{TH}_{1}(x,y)$.
We also use $\gate{MOD}_m(x_1,\ldots,x_d)$ for $m \leq d$ which is defined to be
$1$ if $m$ divides $\sum_{i=1}^d x_i$ and to be $0$ otherwise, and
$\gate{MAJ}(x_1,\ldots,x_d)=\gate{TH}_{d/2}(x_1,\ldots,x_d)$.

A Boolean circuit over the basis
$\set{\gate{NEG},\gate{AND},\gate{OR},\gate{TH}_\theta}$, where all gates
(except for $\gate{NEG}$) may have unbounded fan-in, is called a \emph{threshold
circuit} (TC); we use \gate{AND} and \gate{OR} only for syntactic convenience as
they can be simulated by $\gate{TH}_\theta$.
The problem \TC{}-SAT is, given a threshold circuit~$C$ with exactly one output
gate, to decide whether the circuit is satisfiable, that is, whether there
exists a setting~$x\in\set{0,1}^n$ for the $n$~input gates such that $C(x)=1$.
For $d\in\N$ and $c>0$, a \emph{$c$-sparse-$d$-depth-TC} is a
threshold circuit with $n$ variables, at most~$cn$ wires, and depth at most~$d$.
For each $i\in\N$, a \TC{i}-circuit is a family of threshold circuits of
depth $O(\log^i n)$ and size $\poly(n)$.


\section{Weighted Cliques in Hypergraphs}\label{sec:weightedcliques}

Recall that in the Exact-Weight-$k$-Clique problem on $d$-hypergraphs we are given a $d$-hypergraph $G$ and a target value $t$, and the task is to decide whether some size-$k$ subset $S \subseteq V(G)$ forms a clique of total weight $\sum_{e \in E(G[S])} w(e) = t$.
We denote by~$M=M(w,t)$ the maximum weight in absolute value, that is, we have $M=\max(\{|t|\}\cup\{ |w(e)| :e\in E(G)\})$.
We write $n=|V(G)|$. Since in this section we will mostly deal with the Exact-Weight-$k$-Clique on $d$-hypergraphs problem, we will abbreviate it to ``Weighted $d$-Hypergraph $k$-Clique''.

\subsection{Preprocessing Reductions}
We rely on some basic reductions: The first turns the hypergraph into a complete $d$-hypergraph, which shows that the graph structure is immaterial for this problem; the second makes the hypergraph \mbox{$k$-partite}, which
will be useful in our constructions; the third reduces from ``exact weight
clique'' to ``zero weight clique'', that is, it sets the target value~$t$ to~$0$
by using negative edge weights; the fourth uses a non-negative target
value but removes negative weights.
In the following statement, $M'$ denotes the maximum weight~$M(w',t')$ of the respective output instance.
\begin{fact}\label{fact: clique basic preprocessing}
  Let $d,k\in \N$ with $1\le d\le k$.
  There are $O(n^d)$-time self-reductions for Weighted $d$-Hypergraph
  $k$-Clique with the following properties:
  \begin{enumerate}
    \item ``Make complete'': 
      maps an instance $(G,w,k,t)$ to $(G',w',k,t')$ with $V(G')=V(G)$, $E(G') = \binom{V(G)}{\le d}$, and $M' \le \binom{k}{\le d} M$.
    \item
      ``Make $k$-partite'':
      maps an instance $(G,w,k,t)$ to $(G',w',k,t)$ with $|V(G')|\le
      k|V(G)|$ and ${M=M'}$, such that
      $G'$ is $k$-partite in the sense that $V(G')$ is partitioned into $k$ parts and every edge intersects each part in at most one vertex.
    \item
      ``Make target zero'':
      maps a $k$-partite instance $(G,w,k,t)$ to $(G,w',k,t')$ with $t'=0$
      and $M'\le 2M$.
    \item
      ``Make weights non-negative'':
      maps an instance $(G,w,k,t)$ to $(G,w',k,t')$ with $w':E(G)\to\N$ and $M'\le 2 \binom{k}{\le d}^2 M$.
  \end{enumerate}
\end{fact}
\begin{proof}
  Let $(G,w,k,t)$ be an instance for the problem.
  
  For the first claim, we set $w(e)=\binom{k}{\le d}M$ for edges~$e$ that are supposed to be absent; such edges cannot be used by any solution. Hence, we can assume $E(G)=\binom{V(G)}{\le d}$ without loss of generality.

  For the second claim, we define $V(G')=\{1,\dots,k\}\times V(G)$.
  For every pairwise distinct $a_1,\dots,a_{d'}\in\set{1,\dots,k}$
  and every edge $\set{v_1,\dots,v_{d'}}\in E(G)$ of size~$d'$, we add an edge $f$ with 
  \[f= \set[\big]{(a_1,v_1),\dots,(a_{d'},v_{d'})}\subseteq V(G')\]
  to~$G'$.
  We set the weight $w'(f)=w(\set{v_1,\dots,v_{d'}})$.
  It is clear that this instance is equivalent to the input
  instance, and $k$-partite (the parts consist of vertices with equal first coordinate).

  For the third claim, we slightly modify the weights by setting $t'=0$ and
  subtracting~$t$ from certain edge weights. Specifically, we start with the construction used in the second claim. For any edge of cardinality $d$, denoted by $f=\{(a_1,v_1),\dots,(a_d,v_d)\}$, we set
  $w'(f)=w(\{v_1,\dots,v_d\})$ if $\{a_1,\dots,a_d\}\ne\{1,\dots,d\}$ and
  $w'(f)=w(\{v_1,\dots,v_d\})-t$ if $\{a_1,\dots,a_d\}=\{1,\dots,d\}$. Note that any $k$-clique in~$G'$ contains exactly one edge~$f$ that intersects the first~$d$ parts of the $k$-partition in exactly one vertex each.

  For the fourth claim, we first ensure that
  $E(G)=\binom{V(G)}{\le d}$ using the first claim, which increases $M$ by at most a factor $\binom{k}{\le d}$.
  Let $L=\max\{0, -w(e) : e\in E(G)\}$, that is, $L$
  is the absolute value of the smallest negative weight that occurs in the
  input, or $0$ if there is none.
  We set $w'(e)=w(e)+L$ for all $e$ and $t'=t+L\binom{k}{\le d}$.
  If $t'<0$ or $t'>\max\{w'(e)\}\cdot\binom{k}{\le d}$, the instance is a
  trivial no-instance.
  Otherwise the reduction outputs $(G,w',k,t')$.
\end{proof}

\subsection{Weight Reduction: From Arbitrary to Polynomial}

We proceed by reducing the weights of a given instance of the Weighted $d$-Hypergraph $k$-Clique problem.
By taking the numbers modulo a random prime, we reduce the maximum weight
from~$M$ to~$n^{O(k)}$ in the following way.
\begin{lemma}\label{lemma: random primes}
  Let $d,k\in\N$ with $1\le d\le k$.
  For some constant $f(k,d)\in\N$ there is a randomized $f(k,d)
  \cdot\polylog M$ time self-reduction for the Weighted $d$-Hypergraph $k$-Clique problem that, on input an instance $(G,w,k,t)$ with maximum weight~$M$, makes at most $f(k,d)$ queries to instances $(G,w',k,t')$ where $w':E(G)\to\N$, $t'\in\N$, $M'\le f(k,d)\cdot n^{O(k)}$, and the success probability of the
  reduction is at least $99\%$.
\end{lemma}
\begin{proof}
  If $M\le n^k$ holds, we do not need to do anything.
  If $M \ge \exp(n^k)$ holds, then in time $O(n^k \log M) = \polylog(M)$ we brute-force the problem. 
  In the remaining case, we sample a prime~$p$ uniformly at random from a range specified
  later, and set $w'(e)=w(e)\bmod p$ for all~$e$.
  We query the oracle $(G,w',k,t')$ for all $t'$ with $t'=j p + (t\bmod p)$
  and $j\in \{0,\dots,\binom{k}{\le d} \}$, and we output \emph{yes} if and only
  if at least one oracle query returns \emph{yes}.
  To prove the completeness of this reduction, let~$S$ be a $k$-clique of weight~$t$
  with respect to~$w$.
  Then
  $
  \sum_{e \in E(G[S])} (w(e)\bmod p)
  =jp+\big((\sum_{e \in E(G[S])} w(e))\bmod p\big)
  =jp+ (t \bmod p)= jp+t'
  $
  holds for some $j$ in the specified range since $\binom{k}{\le d}$ is the
  number of terms in the sum.
  Hence yes-instances map to yes-instances with probability~$1$.
  Conversely, if such a $j$ exists, then the weight of $S$ modulo~$p$ is equal
  to $t'$ modulo~$p$.

  For the soundness of the reduction, we need to specify the sampling process for $p$.
  This is implemented as follows: let $Q=200 n^k \log (k^dM)$ and sample positive integers bounded by $O(Q \ln Q)$ uniformly at random until we have found a prime (which we can verify, for example, deterministically in
  time $O(\polylog Q) = O(\polylog M)$ since $M > n$). By the prime number theorem, with probability at least $99.5\%$, after $O(\ln Q) \le O(dk\log n)$ samples we have found a prime that is a uniform sample from a set of at least~$Q$ primes.

  The weight of each $k$-clique~$S$ in~$G$ is at most $k^d M$ in absolute value, and
  there are at most $n^k$ distinct sets~$S$, and so~$n^k$ is also an upper bound
  on the number of distinct weights that appear. For the soundness of the reduction, it is sufficient that $w(S)$ is not congruent to $t$ modulo $p$. As $|t-w(S)|$ is at most $k^d M$, it has at most $\log(k^d M)$ prime divisors. Therefore the probability that for some $S$ it holds that $w(S)$ is congruent to $t$ modulo $p$ is at most $n^k\log(k^d M)/Q = 1/200$.
  Overall, we succeed at finding a prime with the desired property with probability
  $(99.5\%)^2 \ge 99\%$.

  We indeed make at most $k^d$ queries to the oracle, and the largest weight in
  each query is bounded by $\binom{k}{\le d}\cdot p$.
  Since $p$ is bounded by~$Q$ and $M \le \exp(n^k)$, this is at most $f(k,d) n^{O(k)}$.
  For the running time, we need to worry about the bitlength of the involved
  weights.
  The input weights use at most $\log M$ bits, and so~$Q$ (and thus any $p$) uses at most $O(dk \log n + \log M) = \polylog M$ bits.
  Computing the weights modulo~$p$ can be done in time $\polylog M$.
\end{proof}

\subsection{Weight Reduction: From Polynomial to Unweighted}
We reduce the weights from $n^{O(k)}$ to $f(k,d)\cdot \log n$ using a deterministic argument, and then use exhaustive search to reduce to the unweighted case.

The \emph{$q$-expansion} of a number $N\in\N$ is the unique sequence
$N_0,N_1,\dots\in\set{0,\dots,q-1}$ with $N=\sum_{\ell\in\N} N_\ell q^{\ell}$.
After applying the $q$-expansion, all $N_\ell$ are bounded by~$q-1$.
However, the smaller we choose~$q$, the longer the relevant part of the encoding of~$N$ gets; the precise length of this encoding is $p=\lceil\log_q (N+1)\rceil$.
The following lemma uses carries to split the weight constraint along the $q$-expansions of the edge weights~$w$ and the target~$t$.
\begin{lemma}\label{lem: clique weight qexpansion}
  Let $G$ be a $d$-hypergraph with edge-weight function $w:E(G)\to\N$ and a set $S\subseteq V(G)$ with $\abs{S}=k$.
  Let $t,q,p\in\N$ with $q\ge 2$ and $p=\lceil\log_q (t+1)\rceil$.
  The following are equivalent:
  \begin{enumerate}[label=\roman*.]
    \item\label{it: sumtarget}
      (Sum has target value.)
      We have $t=\displaystyle\sum_{e \in E(G[S])} w(e)$.
    \item\label{it: explin}
      (Expansions and carries satisfy linear constraints.)
      There is a sequence $c_0,c_1,\dots\in\{0,\dots,2\binom{k}{\le d}\}$ such that $c_0=0$ and the following linear equations hold for all $\ell\in\N$:
      \begin{align*}
        q c_{\ell+1}
        +
        t_\ell
        =
        c_{\ell}+{\sum_{e \in E(G[S])} w_\ell(e)}
        \,.
      \end{align*}
    \item\label{it: expquad}
      (Expansions and carries satisfy a quadratic equation.)
      There is a sequence $c_0,c_1,\dots\in\{0,\dots,2\binom{k}{\le d}\}$ such that $c_0=0$ and the following quadratic equation holds:
      \begin{equation}\label{eq: weighted clique
          2norm}
        \sum_{\ell\in\N}
        \Big(
          c_{\ell}-t_\ell-q c_{\ell+1}+\sum_{e \in E(G[S])} w_\ell(e)
        \Big)^2=0\,.
      \end{equation}
  \end{enumerate}
\end{lemma}
\begin{proof}
  The equivalence between \ref{it: explin} and \ref{it: expquad} follows from the fact that a sum of squares is zero if and only if all summands are zero.
  To see that \ref{it: sumtarget} implies \ref{it: explin},
  suppose $t=\sum_e w(e)$, so the $q$-expansions are the same as well:
  $t_\ell = \paren[\big]{\sum_e w(e)}_\ell$
  for all $\ell$.
  We inductively set the carries so as to satisfy the linear equations; this choice is unique.
  It remains to argue that the~$c_\ell$ are integers between $0$ and~$2\binom k{\le d}$.
  The fact that they are non-negative integers is a standard property of $q$-expansions, so we only show the upper bound.
  We do so by induction: It clearly holds for~$c_0$.
  In general, we have
  \[
    c_{\ell+1}=-\frac{t_\ell}{q}+\frac{c_\ell}{q}+\sum_e \frac{w_\ell(e)}{q}\,.
  \]
  The first summand~$-t_\ell/q$ is at most~$0$, the second summand~$c_\ell/q$ is at most $2\binom k{\le d}/q\le \binom k{\le d}$ by induction and $q\ge 2$, and the third summand is at most $\binom k{\le d}$, because $w_\ell(e)<q$ holds and the sum has at most $\binom{k}{\le d}$ terms~$e$.

  To see that the second claim implies the first, we observe
  \begin{align*}
    &\sum_{\ell\in\N} q^\ell \sum_{e \in E(G[S])} w_\ell(e)
    =
    \sum_{\ell} q^\ell ( t_\ell +qc_{\ell+1}-c_\ell)\\
    &=
    t + \sum_{\ell>0} q^\ell c_\ell - \sum_\ell q^{\ell} c_\ell
    =
    t-c_0
    =
    t
    \,.
    \qedhere
  \end{align*}
\end{proof}

The following algorithm uses \eqref{eq: weighted clique 2norm} to reduce
weights; in particular, we use the binomial theorem $(a+b)^2=a^2+2ab+b^2$ in
\eqref{eq: weighted clique 2norm} (with $a=c_{\ell}-t_\ell-q c_{\ell+1}$) and then collect terms depending on which
vertices of~$G$ the weight terms depend on -- the terms not depending on edge
weights are collected into the target integer.
As discussed in the introduction, this approach was used before to reduce weights of cliques by Abboud et al.~\cite{AbboudLW14} in the more specific setting of \emph{node weights} in graphs (rather than hypergraps).

\begin{algor}{A}{Weight reduction for the weighted $k$-clique problem}{Given a
    $d$-hypergraph~$G$ with edge weights~$w:E(G)\to\Z$, a number~$k$, a
    weight target~$t\in\Z$, a parameter~$p\in\N$, and access to an oracle for weighted $k$-clique
    in $2d$-hypergraphs, the following algorithm finds a $k$-clique of weight $t$ in~$G$:}
  \item[A1] \emph{(Make $k$-partite and non-negative)}
    Apply Fact~\ref{fact: clique basic preprocessing} to make the instance complete and 
    $k$-partite and all weights non-negative.
  \item[A2] \emph{(Set parameters)}
    Let $M=\max(\{t\}\cup\{w(e):e\in E(G)\})$ and let $q\in\N$ be such that $p=\lceil\log_q M\rceil$.
  \item[A3] \emph{(Guess carries)}
    Exhaustively guess $c_\ell\in\{1,\dots,2\binom k{\le d}\}$ for each
    $\ell\in\{1,\dots,p\}$; set $c_{0}=0$. For each such guess do the following:
    \begin{description}
      \item[a] \emph{(Compute new weights)}
        For every set $f\in\binom{V(G)}{\le 2d}$, let
        \begin{align*}
          w'(f)
          &=
          \sum_{\ell=0}^p
          \Big(
            2\cdot[f\in E(G)]\cdot
            w_\ell(f) \cdot (c_\ell-t_\ell-q c_{\ell+1})
            \\
            &\qquad\qquad+
            \sum_{\substack{e_1,e_2\in E(G)\\e_1\cup e_2 = f}}
            w_\ell(e_1) \cdot w_\ell(e_2)
          \Big)
          \,,\\
          t'&=-\sum_{\ell=0}^{p-1} (c_\ell - t_\ell - q c_{\ell+1})^2
          \,.
        \end{align*}
      \item[b] \emph{(Call oracle)}
        If the oracle detects a $k$-clique~$S$ in $(G',w')$ of weight~$t'$, then halt and output \emph{yes}; otherwise continue guessing carries.
    \end{description}
  \item[A4]
    If no suitable carries were found, output \emph{no}.
\end{algor}

\begin{lemma}\label{lemma: clique weight reduction}
  Let $d,k\in\N$ with $1\le d\le k$.
  Algorithm A (with input parameter~$p\in\N$) is an oracle reduction from
  Weighted $d$-Hypergraph $k$-Clique to Weighted $2d$-Hypergraph $k$-Clique.
  The algorithm runs in time $O(p4^dn^{2d}k^{dp})$ and makes at most $k^{dp}$ oracle
  queries.
  Every query is a hypergraph on the same set of vertices.
  If $M$ is the maximal weight among $w$ and $t$, then the maximal weight~$M'$
  of all queries satisfies $M'\le O(k^{4d} M^{2/p} p)$.
\end{lemma}
\begin{proof}
  Let $G$ be the $d$-hypergraph with $E(G)=\binom{V(G)}{\le d}$, edge weight
  function $w:E(G)\to\N$, and target~$t\in\N$ after applying
  Fact~\ref{fact: clique basic preprocessing}.
  Let $G'$ be the $2d$-hypergraph with $E(G')=\binom{V(G)}{\leq 2d}$.

  We first prove the correctness of the reduction.
  By Lemma~\ref{lem: clique weight qexpansion}, the instance $(G,k,w,t)$
  has a $k$-clique~$S$ of total weight $t$ if and only if there
  exist $(c_\ell)_\ell$ satisfying~\eqref{eq: weighted clique 2norm}.
  Now consider the weight of~$S$ in~$G'$.
  We abbreviate the terms in the left side of \eqref{eq: weighted clique 2norm} with
  $a_\ell=c_\ell-t_\ell-qc_{\ell+1}$ and $b_\ell=\sum_e w_\ell(e)$, and have
  \[
    \sum_{\ell\in\N} (a_\ell + b_\ell)^2
    =
    \sum_{\ell\in\N} \paren[\big]{a^2_\ell + 2 a_\ell b_\ell+ b^2_\ell}
    =
    \sum_{\ell\in\N} a^2_\ell
    +
    \sum_{\ell\in\N} \paren[\big]{2 a_\ell b_\ell+ b^2_\ell}
    \,.
  \]
  The squares of~$b_\ell$ expand as follows:
  \begin{align*}
    b_\ell^2
    =
    \paren[\big]{\sum_{e} w_\ell(e)}^2
    =
    \sum_{\substack{e_1,e_2\in E(G)\\e_1\cup e_2 = f}}
    w_\ell(e_1) \cdot w_\ell(e_2)
    \,.
  \end{align*}
  Now we observe that $t'$ and $w'(f)$ were defined exactly as to satisfy
  \begin{align*}
    \sum_{\ell\in\N} (a_\ell + b_\ell)^2
    =
    -t' + \sum_{f \in E(G'[S])} w'(f)\,.
  \end{align*}
  By Lemma~\ref{lem: clique weight qexpansion},
  the set $S$ is a $k$-clique of weight~$t$ in~$(G,w)$ if and only if the right side of
  the latter equation is equal to~$0$, which in turn holds if and only if the
  weight of~$S$ with respect to~$w'$ is~$t'$.

  For the running time, note that the preprocessing takes $O(n^d)$ time.
  Exhaustive search for the carries takes $O(k^{dp})$ iterations, and each
  iteration takes time $O(n^{2d}p4^d)$ because of line A3a in which we need to compute $w'(f)$ for every edge; overall the reduction takes time
  $O(n^{2d}k^{dp})$ and makes at most $k^{dp}$ oracle queries.

  For the weights, note $c_\ell\le 2 k^d$ and so
  $|t'|\le O(k^{2d} q^2 p)$.
  To get the bound on $w'(f)$, observe that the term
  \[\sum_{e_1\cup e_1=f} w_\ell(e_1)w_\ell(e_2)\] is bounded by $4^dq^2$, and the term
  $w_\ell(f)\cdot(c_\ell-t_\ell-q c_{\ell+1})$ is bounded by $O(q^2k^d)$ in
  absolute value.
  The preprocessing step relying on Fact~\ref{fact: clique basic preprocessing}
  may have added an additional factor of~$k^{2d}$; overall, all weights are bounded
  by $O(k^{4d}q^2p)$.
\end{proof}

We apply Lemma~\ref{lemma: clique weight reduction} to reduce the maximum
weights from $\poly(n)$ to $O(\log n)$, which is small enough to allow for 
exhaustive enumeration to reduce to the problem without weights.
\begin{lemma}\label{lemma: clique weight loss}
  Let $d,k\in\N$ with $1\le d\le k$ and $f(d,k)\in\N$.
  There is an $n^{2d+o(1)}$-time oracle reduction from Weighted $d$-Hypergraph
  $k$-Clique with weights in $\{-n^{f(k,d)},\dots,n^{f(k,d)}\}$ to unweighted
  $k$-partite $2d$-hypergraph $k$-Clique.
  If the input has~$n$ vertices, every oracle query has~$n$ vertices and the reduction uses at most $n^{o(1)}$ queries. Here the $o(1)$ terms are of the form $g(k,d) / \sqrt{\log n}$.
\end{lemma}
\begin{proof}
  Let $G$ be a $k$-partite $d$-hypergraph with edge-weight function~$w:E(G)\to\Z$ and target value~$t\in\Z$.
  We apply Lemma~\ref{lemma: clique weight reduction}.
  In particular, setting $p = \sqrt{\log n}$,
  we get
  $k^{dp}=n^{o(1)}$ queries and maximum weight
  $M'\le O(k^{4d} M^{2/p} p) = n^{o(1)}$.

  Each query is now a $k$-partite instance $(G',w',k,t')$ with maximum
  weight~$M'$, where we treat $k$ and $d$ as constants.
  A solution~$S$ of $G'$ satisfies $\sum_{e \in E(G'[S])} w'(e)=t'$.
  Since~$G'$ is $k$-partite,~$S$ intersects each part in exactly one vertex, and
  for each set~$C\subseteq\{1,\dots,k\}$ with $1\le|C|\le d$, there is a unique
  edge $e_C \in E(G'[S])$ that intersects exactly the color classes in~$C$, and
  this edge contributes $w'(e_C)$ to the total weight of~$S$.
  We want to simulate these weights by exhaustively guessing the contribution $w'(e_C)$ of
  each~$C$. To do so, we only keep the edges of color type~$C$ that have the
  guessed weight.

  More precisely, for each $C\subseteq\{1,\dots,k\}$ with $1\le |C|\le d$, we
  exhaustively guess a weight $a_C\in[-M',M']$.
  In total, this requires iterating through at most $(2M'+1)^{k^d}=n^{o(1)}$ candidate weight vectors~$a=(a_C)_{C\subseteq\set{1,\dots,k}}$.
  If the sum $\sum_C a_C$ is not equal to~$t'$, we reject the candidate vector and move to the next one.
  Otherwise, for each~$C$ and each edge~$e$ intersecting exactly the color
  classes prescribed by~$C$, we keep $e$ in the graph if and only
  if~$w'(e)=a_C$.
  In this way, we obtain a~$k$-partite $2d$-hypergraph~$G_a$.
  For each candidate vector~$a$ of weight~$t'$, we query the (unweighted) $k$-clique oracle for $2d$-hypergraphs and output yes if and only if at least one query outputs yes.

  The claim on the running time follows, since there are only $n^{o(1)}$ candidate vectors when~$k$ is regarded as a constant, and each oracle query~$G_a$ is prepared in time~$n^{2d+o(1)}$, which is almost-linear in the description length of~$G_a$.
  For the correctness, note that $G$ has a solution~$S$ if and only if~$G'$ has a
  solution~$S$.
  If~$G'$ has a solution~$S$, then there is a setting of the~$a_C$ corresponding
  to the solution such that all edges in~$G[S]$ survive in~$G_a$, and the oracle
  finds a $k$-clique.
  On the other hand, if~$S$ is a $k$-clique in some~$G_a$, then the used edges
  have the desired weight in~$G'$.
  The correctness of the reduction follows.
\end{proof}

\subsection{Reduction to Orthogonal Vectors}

In this section, we reduce from $k$-Clique in $d$-hypergraphs via the $k$-OV problem to $2$-OV.
Recall that the $k$-OV problem is, given~$k$ sets $X_1,\dots,X_k\subseteq\{0,1\}^D$ of
Boolean vectors, to find $x_1\in X_1$, \ldots, $x_k\in X_k$ such that $\sum_{j=1}^D \prod_{i=1}^k x_{ij} = 0$ holds, where the sum and product are the usual operations over the integers.
\begin{lemma}\label{lemma: kclique to kOV}
  Let $d,k\in\N$ with $1\le d\le k$.
  There is a many-one reduction from (unweighted)
  $k$-partite $d$-hypergraph $k$-Clique to $k$-OV that runs in time $O(n^{d+1})\polylog n$;
  the number of produced vectors is $n$ and the dimension of the vectors
  is~$n^d$.
\end{lemma}
\begin{proof}
  Let $G$ be a $k$-partite $d$-hypergraph with parts $V_1,\ldots,V_k$.
  Let $v_1,\dots,v_k$ be vertices with $v_i\in V_i$ for all $i\in\set{1,\dots,k}$.
  Then $\set{v_1,\dots,v_k}$ forms a $k$-clique in~$G$ if and only if all non-edges~$h\in\overline E(G)$ satisfy $h\not\subseteq\{v_1,\dots,v_k\}$.
  Here $\overline E(G)$ denotes the set
  \[\setc*{e \in \binom{V(G)}{\leq d}}{\forall i. |e \cap V_i| \leq 1}\setminus E(G)\,.\]

  We construct the instance $X_1,\dots,X_k$ of~$k$-OV as follows.
  For each $v\in V_i$, we create a vector $x_v\in X_i\subseteq\set{0,1}^{\overline E(G)}$ as follows:
  If $h\in\overline E(G)$ is disjoint from $V_i$, we set $x_{v,h}=1$.
  If $h\cap V_i = \set{v}$, we set $x_{v,h}=1$.
  Otherwise we have $h\cap V_i =\set{u}\ne \set{v}$ for some~$u$, and we set $x_{v,h}=0$.
  Clearly the sets~$X_1,\dots,X_k$ contain a total of~$n$ Boolean vectors, each with $\abs{\overline E(G)}\le n^d$ dimensions.
  Moreover, the sets are easily computed in $O(n^{d+1}\polylog n)$ time.
  It remains to prove the correctness of the reduction.

  To this end, let $v_1,\dots,v_k$ with $v_i\in V_i$ for all~$i$ be vertices that form a $k$-clique $\set{v_1,\dots,v_k}$ in the $k$-partite $d$-hypergraph~$G$.
  We claim that $\set{x_{v_i}}$ is a solution to the~$k$-OV instance, that is,
  we claim $\sum_{h\in\overline E(G)} \prod_{i=1}^k x_{{v_i},h}=0$.
  To see this, let $h\in\overline E(G)$ be arbitrary.
  Since $\set{v_1,\dots,v_k}$ is a $k$-clique in~$G$ and~$h$ is a non-edge of~$G$, there exists a part~$V_i$ that satisfies $V_i\cap h\ne\emptyset$ and $v_i\not\in h$.
  By definition, we have $x_{v_i,h}=0$.
  Thus the entire sum is indeed zero.

  For the reverse direction, let $x_{v_1},\dots,x_{v_k}$ with $x_{v_i}\in X_i$ for all $i$ be vectors that form a solution to the $k$-OV instance.
  This means that for all $h\in\overline E(G)$, there exists some $i\in\set{1,\dots,k}$ such that $x_{v_i,h}=0$ holds.
  By definition, this implies that $h\cap V_i = \set{u}\ne\set{v_i}$ for some~$u$ holds.
  Thus in particular, $h\not\subseteq\set{v_1,\dots,v_k}$ and so the set $\set{v_1,\dots,v_k}$ does not contain any non-edges of~$G$ and must be a clique.
\end{proof}

The last step of our reduction is reminiscent of the classic SETH-based lower bound for the 2-OV problem~\cite{Wil05}.

\begin{lemma}\label{lemma: kOV to 2OV}
  Let $k\in\N$.
  There is an $O(n^{\lceil k/2\rceil}D)$ time many-one reduction from $k$-OV 
  to $2$-OV that maps instances with $n$ vectors in dimension $D$ to instances with
  $O(n^{\lceil k/2 \rceil})$ vectors in dimension~$D$.
\end{lemma}
\begin{proof}
  Let $X_1,\dots,X_k \subseteq \{0,1\}^D$ be the input for $k$-OV with $n=\sum_{i=1}^k|X_i|$.
  The idea is to split the instance into two halves and list all candidate solutions in each half.
  For each candidate solution $S\subseteq X_1\cup\dots\cup X_{\lfloor k/2 \rfloor}$ with
  $|S\cap X_i|=1$ for all $i\in\{1,\dots,\lfloor k/2\rfloor\}$, we create a
  vector $v^S\in X'_1\subseteq\{0,1\}^D$ by setting $v^S_i=\prod_{u\in S} u_i$.
  Similarly, for each $S'\subseteq X_{\lfloor k/2 \rfloor+1}\cup\dots\cup X_k$ with
  $|S'\cap X_i|=1$ for all $i\in\{\lfloor k/2\rfloor+1,\dots,k\}$, we create a
  vector $v^S\in X'_2\subseteq\{0,1\}^D$ by setting $v^S_i=\prod_{u\in S} u_i$.
  We obtain an instance $X'_1,X'_2\subseteq\{0,1\}^D$ of $2$-OV.

  We claim that $X_1,\dots,X_k$ is a yes-instance of $k$-OV if and only if $X'_1,X'_2$ is a yes-instance of~$2$-OV.
  Suppose that $v_1,\dots,v_k$ are orthogonal, that is, $\prod_{i=1}^k (v_i)_j=0$
  holds for all $j\in\{1,\dots,D\}$.
  We set $S=\{v_1,\dots,v_{\lfloor k/2 \rfloor}\}$
  and $S'=\{v_{\lfloor k/2 \rfloor+1},\dots,v_k\}$.
  Clearly $v^S_j\cdot v^{S'}_j=0$ holds for all $j$, so $v^S,v^{S'}\in V'$ are
  indeed orthogonal.
  Conversely, if $v^S_j \cdot v^{S'}_j=0$ holds for all $j$, then the $k$ vectors in
  $S\cup S'$ are orthogonal.
\end{proof}

\subsection{Tying Things Together}

We now formally prove Theorem~\ref{theorem: OV main}.

\begin{reptheorem}{theorem: OV main}
  \theoremOVmain
\end{reptheorem}
\begin{proof} Let ($G,w,k,t)$ be an instance of Min-Weight-$k$-Clique on $d$-hypergraphs. We summarize the lemmas of this section as follows:
  
\begin{enumerate}
	\item Lemma~\ref{lemma: random primes} randomly reduces in $\polylog(M)$ time from Exact-Weight-$k$-Clique with
  weights up to $M$ to a constant (which depends on $k$ and $d$) number of instances of Exact-Weight-$k$-Clique on $G$ with weights up to
  $n^{O(k)}$. 
	\item Lemma~\ref{lemma: clique weight loss} reduces this in $n^{2d+o(1)}$ time to $n^{o(1)}$ instances of $k$-Clique on $2d$-hypergraphs.
	\item Lemma~\ref{lemma: kclique to kOV} reduces in $n^{2d+1+o(1)}$ time any instance of $k$-Clique on $2d$-hypergraphs to an instance of $k$-OV with $n$ vectors in $n^{2d}$ dimensions.
	\item Lemma~\ref{lemma: kOV to 2OV} reduces in time $n^{\lceil k/2 \rceil+2d+o(1)}$ any such an instance of $k$-OV to an instance of $2$-OV with $O(n^{\lceil k/2 \rceil})$ vectors in $n^{2d}$ dimensions.
\end{enumerate}
  Composing the reductions gives a randomized $O(n^{\lceil k/2 \rceil+2d+o(1)})+ \polylog(M)$ time oracle reduction from Exact-Weight-$k$-Clique on $d$-hypergraphs with $n$ vertices and largest weight $M$ to 2-OV on~$n^{\lceil k/2 \rceil}$ vectors of dimension $n^{2d}$ using $n^{o(1)}$ oracle calls. 
  
  To reduce from Min-Weight-$k$-Clique to Exact-Weight-$k$-Clique we use
  \cite[Theorem~1]{DBLP:conf/mfcs/NederlofLZ12}, which allows us to perform a
  binary search for the minimum-weight-$k$-clique by making few queries to
  Exact-Weight-$k$-Clique. As the domain $E(G)$ of our weight function is of size $O(n^d)$ and the maximum weight is upper bounded by $M$, this reduction requires $O(n^{d} \log M)$ oracle calls. 
  This yields a
  randomized
  oracle reduction from Min-Weight-$k$-Clique on $d$-hypergraphs with~$n$ vertices and largest weight $M$ to 2-OV,
  which
  runs in time $O(n^{\lceil k/2 \rceil+2d+o(1)}) \polylog(M)$
  and
  makes $n^{d+o(1)} \log M$ oracle calls.
  To reduce from $k$-Clique on $d$-hypergraphs with~$n$ vertices to 2-OV, we only use steps 3 and 4, which takes time $O(n^{\lceil k/2 \rceil+2d+o(1)})$.

  All three reductions produce instances of 2-OV with $N$ vectors in~$D$ dimensions, where $N=n^{\lceil k/2 \rceil}$ and $D=n^{2d}=N^{\delta}$ for some $\delta=\delta(k,d)>0$.
  If $k$ is large enough compared to $d$, then $\delta$ is arbitrarily close to zero.
  If the moderate-dimension OV conjecture is false, there exist $\delta,\eps'>0$ such that 2-OV with $D=N^\delta$ has an $O(N^{2-\eps'})$ time algorithm.
  Combined with any of the three reductions, we obtain three algorithms that run in time at most
  \[
    O\paren[\Big]{n^{\lceil k/2 \rceil+2d+o(1)} +n^{\lceil k/2 \rceil(2-\eps')+d+o(1)}}\polylog(M)
    \,.
  \]
  When $k$ is large enough compared to~$d$, this is $O(n^{k(1-\eps)}) \polylog(M)$ for some $\eps=\eps(\eps') > 0$.
  This proves the claim.
\end{proof}

And we have the following corollary to Theorem~\ref{theorem: OV main}.

\begin{repcorollary}{cor: maxsat}
  \corollarymaxsat
\end{repcorollary}    
    
The corollary follows from Theorem~\ref{theorem: OV main} with a reduction from Max-$d$-SAT to $k$-Clique on $d$-hypergraphs that was already sketched in e.g.~\cite{Wil05}.
We formally state and prove this reduction next. 

\begin{lemma} \label{lem:cliquemaxdsat}
  Let $d,k\in\N$ with $1\le d\le k$.
  There is an $O^*(2^{dn/k})$ time reduction from Max-$d$-SAT to Min-Weight-$k$-Clique on $d$-hy\-per\-graphs
  that maps $d$-CNF formulas with $n$ variables and $m$ clauses
  to $d$-hypergraphs with at most $k2^{n/k}$ vertices and integer edge weights between $-2m$ and $2m$.
\end{lemma}
\begin{proof}
Given an instance of Max-$d$-SAT consisting of a $d$-CNF formula $\varphi$ on variable set $V$ of size $n$ and $m$ clauses, and an integer~$t$ indicating the required number of satisfied clauses, partition $V$ into sets $V_1,\ldots,V_k$ where $|V_k| \leq n/k$.
The reduction computes from~$\varphi$ an instance~$H$ of Min-Weight-$k$-Clique.

We build a complete $k$-partite $d$-hypergraph $H$ with vertices $\bigcup_{i=1}^k P_i$ where $P_i$ contains a vertex $p^{i}_x$ for every vector $x \in \{0,1\}^{V_i}$.
Create an edge $f=\{x_1,\ldots,x_\ell\}$ for every set $\{i_1,\ldots,i_\ell\} \in \binom{[k]}{\leq d}$ and tuple $(x_1,\ldots,x_\ell) \in P_{i_1} \times \ldots \times P_{i_\ell}$.
Define the weight of~$f$ to be $-1$ times the number of clauses that
\begin{enumerate}
	\item are contained in $\bigcup_{j=1}^{\ell}V_{i_j}$,
	\item contain a variable in $V_{i_j}$ for every $j=1,\ldots,\ell$, and
	\item are satisfied by the partial assignment obtained by setting the variables in $V_{i_j}$ according to $x_{j}$.
\end{enumerate}
The target instance is $H$, and the goal is to decide whether the minimum weight of any $k$-clique is at most $-t$. As the number of edges of $H$ is at most $(k2^{n/k})^d$ and we compute their weights in polynomial time, the running time of this reduction is bounded by $O^*(2^{dn/k})$. 

To see that this is a correct reduction, let $X \subseteq V(H)$ be a $k$-clique of $H$ of weight at most $-t$. We see that $X$ contains at most one vertex from every $P_{i}$, and as $|X|=k$ we have that $X$ intersects in exactly one vertex with every $P_i$. By definition of the sets $P_i$, the set $X$ thus corresponds to an assignment $x$ of the variables of $\varphi$. We claim that the weight of $X$ is $-1$ times the number of clauses satisfied by $x$ and therefore $\varphi$ has an assignment satisfying at least $t$ clauses. To see the claim, let $C$ be a clause of~$\varphi$ and $\{i : C \text{ contains variables from } V_i \} = \{i_1,\ldots,i_\ell\}$ be the set of variable groups intersecting~$C$. Let $x_{i_1},\ldots,x_{i_\ell}$ be the corresponding partial assignments that $x$ induces to $V_{i_1},\ldots,V_{i_\ell}$. We see that~$C$ contributes~$-1$ to the weight of the hyperedge $(x_1,\ldots,x_\ell)$ and $0$ to all other edges.

For the reverse direction, suppose $x$ is an assignment satisfying at least $t$ clauses of $\varphi$ and let $x_i$ be the projecting of $x$ onto $V_i$. Then by the above claim the weight of $X := \{p^1_{x_1},\ldots,p^k_{x_k}\}$ is $-t$.
\end{proof}

\begin{proof}[Proof of Corollary~\ref{cor: maxsat}]
  For a sufficiently large constant~$k$, we combine the reduction in Lemma~\ref{lem:cliquemaxdsat} with an $O(n^{(1-\eps)k})\polylog M$ time algorithm for Min-Weight-$k$-Clique in $d$-hypergraphs.
  This yields an $O^*(2^{(1-\eps)n})$ time algorithm for Max-$d$-SAT.
  Together with Theorem~\ref{theorem: OV main} this proves the claim.
\end{proof}


\section{Reducing Sparse Satisfiability Problems to CNF-SAT}

A dream theorem would be to reduce the sparse circuit satisfiability problem
over the De Morgan basis to the CNF-SAT problem in such a way that a violation
of SETH implies that faster algorithms for sparse circuit satisfiability exist.
We demonstrate how to do this in Section~\ref{sec: formulas} for sparse formulas as a warm-up, reproving a result of~\cite{DBLP:conf/ciac/DantsinW13}.
In Section~\ref{sec: TC0}, we prove Theorem~\ref{theorem: TC0}, the extension of this result to sparse \TC0-circuits.
We also prove that Theorem~\ref{theorem: TC1}) for sparse \TC1-circuits and CNF-SAT in Section~\ref{sec: TC1}.

\subsection{Sparse Formulas}\label{sec: formulas}
\emph{Formulas} are circuits that become a tree when the input gates are
removed.
We consider formulas over the De Morgan basis
$\set{\gate{NEG},\gate{AND},\gate{OR}}$; in particular, we assume the
corresponding trees to be binary, that is, the fan-in of every gate is at most
two.
We use the following simple decomposition lemma for binary trees:
\begin{lemma}%
[Impagliazzo, Meka, and Zuckerman {\cite[Claim 4.4]{DBLP:conf/soda/ImpagliazzoMP12}}]
  \label{lemma: formula decomposition}
  Let $T$ be a binary tree with $m$ nodes and let $\ell\in\N$ with
  $\ell\le m$.
  There exists a set $A\subseteq V(T)$ with $\abs{A}\le 6m/\ell$ such that every
  connected component~$C\subseteq V(T)$ of $T-A$ has at most $\ell$ nodes and
  at most three vertices of~$A$ are adjacent to vertices of~$C$.
  Moreover, such a set~$A$ can be computed in polynomial time.
\end{lemma}
Using this lemma, the satisfiability of sparse Boolean
formulas reduces to the satisfiability of $k$-CNF formulas with only a small
overhead in the running time.

\begin{algor}{B}{Transform sparse formula to $k$-CNF}{Given a Boolean
  formula~$F$ and an positive integer~$k\in\N$, this algorithm computes
  an equivalent $k$-CNF formula~$F'$.}
  \item[B1] \emph{(Compute decomposition)} Let $A\subseteq V(F)$ be the set
    guaranteed by Lemma~\ref{lemma: formula decomposition} where $\ell=k/4$ and
    $T=F-I(F)$ is the tree obtained from~$F$ by removing its input gates.
  \item[B2] \emph{(Create variables)}
    Let $x_1,\dots,x_n$ be the input variables of~$F$; for each $a\in A$, create
    a variable~$y_a$.
    Also create a variable~$y_o$ where~$o\in V(F)$ is the output gate of~$F$. 
  \item[B3] \emph{(Compute $k$-CNFs for small subcircuits)}
    For each $v\in A\cup\set{o}$, do the following:
    \begin{description}
      \item [a] Let $F_{v,A}$ be the subcircuit of~$F$ induced by the
        set~$R_C(A,v)$ (see (\ref{eq:defRCAv}) in the preliminaries) and interpret the gates~$a\in A$ as input
        variables~$y_a$.
      \item [b] Since $F_{v,A}$ depends on at most $2\ell$ variables, we can compute
        a $k$-CNF formula $F'_v$ with at most $2\ell+1\le k$ variables that
        expresses the constraint ``$y_v = F_{v,A}(x,y)$''.
    \end{description}
  \item[B4] \emph{(Output)}
    Let $F'= y_o\wedge \bigwedge_{v\in A} F'_v$, and output~$F'$.
\end{algor}

We prove the correctness and properties of this algorithm.
\begin{lemma}
  \label{lemma: formula reduction}
  Let $c,\eps>0$.
  There exists a $k\in\N$ such that algorithm B is a polynomial-time
  many-one reduction for the satisfiability problem that maps formulas with $n$
  variables and at most $cn$ gates to a $k$-CNF formula with at most
  $(1+\eps)\cdot n$ variables.
\end{lemma}
\begin{proof}
  We set $k=c/(24\eps)$.
  Let $F$ be the input formula with~$n$ variables~$x_1,\dots,x_n$ and $m\le cn$ gates.
  We claim that $F'$ has a satisfying assignment if and only if~$F$ does,
  and~$F'$ is a $k$-CNF formula with at most $(1+\eps)n$ variables.

  Let $T$ be the tree obtained when removing the input gates, and let
  $A\subseteq V(T)$ be the vertex set guaranteed by Lemma~\ref{lemma: formula
  decomposition} for $\ell=k/4$.
  Since $m\le cn$, we have $\abs{A}\le 24 cn/k \le \eps n$, so indeed~$F'$ has
  at most $(1+\eps)n$ variables.
  By Lemma~\ref{lemma: formula decomposition}, $F_{v,A}$ contains at most~$\ell$
  non-input gates and, since the fan-in is at most two, $F_{v,A}$ contains most
  $2\ell$ gates overall.
  So the constraint ``$y_v = F_{v,A}(x,y)$'' indeed depends on at most $2\ell+1$
  variables and can be expressed trivially by a $(2\ell+1)$-CNF formula.
  It is clear that~$F'$ can be computed in polynomial time.
  Moreover, $F(x)=1$ holds if and only if there is a setting for~$y$ such that
  $F'(x,y)=1$ holds, so $F$ and $F'$ are equisatisfiable.
  We obtain the claimed reduction.
\end{proof}

Using this lemma, we prove that SETH is implied by an analogue of SETH for sparse formulas.
\begin{theorem}[Dantsin and Wolpert \cite{DBLP:conf/ciac/DantsinW13}]\label{theorem: sparse formula SETH}
  If SETH is false, then there is an $\eps>0$ such that, for all $c$, the
  satisfiability of Boolean formulas of size at most $cn$ can be solved in time
  $O\paren[\bi]{(2-\eps)^n}$.
\end{theorem}

\begin{proof}
  Suppose that SETH is false.
  Then there is some $\delta>0$ such that $k$-CNF-SAT can be solved in time
  $O^*((2-\delta)^n)$ for all $k\in\N$.
  Let $c>0$.
  In order to solve Formula-SAT for $cn$-size formulas, we apply
  Lemma~\ref{lemma: formula reduction} to reduce to a $k$-CNF formula with
  $n'=(1+\alpha)n$ variables.
  We can solve this instance using the assumed algorithm in time
  $O((2-\delta)^{n'}) = O((2-\delta)^{(1+\alpha)n}) = O((2-\eps)^n)$ for some suitable
  $\eps,\alpha>0$.
\end{proof}

\subsection{Sparse TC0-circuits}\label{sec: TC0}
The goal of this section is to prove the following:

\begin{lemma}\label{lem:rewrite}
  There is a polynomial-time many-one reduction from TC-SAT to CNF-SAT that,
  given $\eps \in (0,1)$ and a depth-$d$ threshold circuit with at most $cn$ wires, with $c \ge 1$,
  produces a $k$-CNF formula $\varphi$ with at most $(1+\eps)n$ variables and width
  $k \leq \paren[\big]{2000(c/\eps)\log(2c / \eps)}^d$.
\end{lemma}
Our proof of Lemma~\ref{lem:rewrite} relies on a linear-size adder circuit as
provided by the following lemma.
\begin{lemma}[Adder circuit]\label{lem:adder}
  Let $b,\ell\in\N$.
  There is a circuit $C_{add}:\set{0,1}^{b \ell}\to\set{0,1}^{b+\lceil \log \ell \rceil}$ over
  $\set{\gate{NEG},\gate{AND},\gate{OR}}$ with at most $40b \ell$ gates and maximum
  fan-in $2$ such that $C_{add}$ computes the binary representation of the sum
  of $\ell$ given $b$-bit integers.
\end{lemma}
\begin{proof}
  The proof of this lemma is standard.
  It is well known that there is a circuit $C_{FA}$ (the full adder) that adds $b'$-bit numbers using $20b'$ gates and constant fan-in.
  Describing the computation of $C_{add}$ using parentheses, the circuit $C_{add}$ computes the sum $\sum_{i=1}^\ell b_i$ in a binary-tree-like way as 
  \[
    (((b_1+b_2)+(b_3+b_4))+((b_5+b_6)+(b_7+b_8))) + \ldots
  \] 
  Using $C_{FA}$ for every addition, the number of gates needed is at most 
  \[
    \sum_{i=1}^{\lceil \lg \ell \rceil } 20(b+i-1) \ell / 2^i \leq 20b \ell \sum_{i=1}^{\infty} i/2^i = 40b \ell \;. \qedhere
  \]
\end{proof}

Our proof also needs a circuit computing the threshold function for binary inputs.

\begin{lemma}[BINTH circuit]\label{lem:binth}
  Let $r,\theta\in\N$.
  There is a circuit $\gate{BINTH}_{\theta}:\set{0,1}^{r}\to\set{0,1}$ over
  $\set{\gate{NEG},\gate{AND},\gate{OR}}$ with at most $2r$ gates and maximum
  fan-in $2$ such that $\gate{BINTH}_{\theta}(x_0,\ldots,x_{r-1})$ computes whether $\sum_{i=0}^{r-1} x_i 2^{i}$ is at least $\theta$.
\end{lemma}
\begin{proof}
        The circuit $\gate{BINTH}_\theta$ is constructed by
        setting~$t=\lceil \lg \theta \rceil$ and converting the following
        expression to a circuit:\\
        $
        \gate{BINTH}_{\theta}(x_0,\ldots,x_{r-1})
        =
        $
        \begin{align*}
          \paren[\Big]{
          \bigvee^{r-1}_{i=t} x_i
          }
          \vee
          \paren[\Big]{
            x_{t-1} \wedge \gate{BINTH}_{\theta-2^{t-1}}(x_0,\ldots,x_{t-2})
            }
          \,.
          \mbox{}&
          \qedhere
        \end{align*}
\end{proof}

We now have all tools needed to prove Lemma~\ref{lem:rewrite}.

\begin{proof}[Proof of Lemma~\ref{lem:rewrite}.]
  Without loss of generality, threshold circuits only have input, $\gate{NEG}$,
  and $\gate{TH}_\theta$ gates, since
  $\gate{TH}_\theta$ can directly simulate $\gate{AND}$, $\gate{OR}$, and
  $\gate{MAJ}$ gates.
  The algorithm to transform threshold circuits into $k$-CNF formulas is
  implemented in Algorithm~C.
  Intuitively, it replaces threshold gates of large fan-in by a circuit of
  bounded fan-in in such a way that the circuit can be simulated by a~$k$-CNF
  formula without introducing too many new variables.

\begin{algor}{C}{Reduce sparse threshold circuit to $k$-CNF}{Given a threshold
    circuit~$C$ of depth~$d$ and a positive integer~$\beta\in\N$, this
    algorithm computes an equivalent $k$-CNF formula~$F$.}
  \item[C1] \emph{(Initialize gates to be replaced with variables)}
    Let $A=\set{o}$ where $o\in V(C)$ is the output gate of~$C$.
  \item[C2] \emph{(Replace large threshold gates by the circuit in Figure~\ref{fig:thgate})}
    For each $v\in V(C)$ of with fanin~$d^-_C(v) \ge \beta$:
    \begin{description}
      \item[a]
        Let $\theta\in\N$ such that $v$ is a $\gate{TH}_\theta$-gate.
      \item[b]
        Partition the children $N_C^-(v)$ of~$v$ into blocks~$B_1,\dots,B_\ell$
        of size at most~$\beta$ with $\ell\le \lceil d^-_C(v)/\beta \rceil$ and
        remove all wires leading into~$v$.
      \item[c] \emph{(Construct adder circuit for each block)}
        For each $i\in\set{1,\dots,\ell}$, create a circuit~$C_{add}(B_i)$ that
        uses the gates of~$B_i$ as input gates and has $\log\beta+1$ output
        gates~$b_i$ such that $b_i$ represents the number of~$1$s in $B_i$ in
        binary. Here, $C_{add}$ is the circuit from Lemma~\ref{lem:adder}.
      \item[d] \emph{(Simulate threshold gate by circuit of fan-in two)}
        Add circuit~$\gate{BINTH}_\theta(C_{add}(b_1,\dots,b_\ell))$ with
        inputs~$b_1,\dots,b_\ell$ and output gate~$v$, that is, the inputs $b_1,\dots,b_\ell$ 
        are fed into~$C_{add}$ (from Lemma~\ref{lem:adder}), whose outputs are fed 
        into~$\gate{BINTH}_\theta$ (from Lemma~\ref{lem:binth}).
        This concatenated circuit 
        takes the binary representation of~$\ell$
        integers~$(b_1,\dots,b_\ell)\in\set{0,1,\ldots,\beta}^\ell$ on~$b=\log\beta+1$ bits each and it
        outputs true if and only if the sum of the given integers is at
        least~$\theta$.
        The circuit~\[\gate{BINTH}_\theta(C_{add}(b_1,\dots,b_\ell))\] has at
        most~$40b\ell + 2(b + \lceil \log \ell \rceil ) \le 44 b \ell$ gates and fan-in at most~$2$.
        We add all of its
        gates, including the~$b_i$'s, to~$A$.
    \end{description}
  \item[C3] \emph{(Compute $k$-CNFs)}
    For all $v\in A$, add a new variable~$y_v$ and do the following:
    \begin{description}
      \item [a] Let $C_{v,A}$ be the subcircuit of~$C$ induced by the
        set~$R_C(A,v)$ (see (\ref{eq:defRCAv}) in the preliminaries) and interpret the gates~$a\in A$ as input
        variables~$y_a$.
      \item [b] We will show that $C_{v,A}$ depends on at most $\beta^d$
        variables, so we can compute a $k$-CNF formula $F_v$ with at most
        $k=\beta^d$ variables that expresses the constraint ``$y_v =
        C_{v,A}(x,y)$'', where $x_1,\dots,x_n$ are the input variables of~$C$.
    \end{description}

  \item[C4] \emph{(Output)} Let $F=y_o \wedge \bigwedge_{v\in A} F_v$, and output~$F$.
  \end{algor}

\begin{figure}
	\centering
	\includegraphics[scale=0.6]{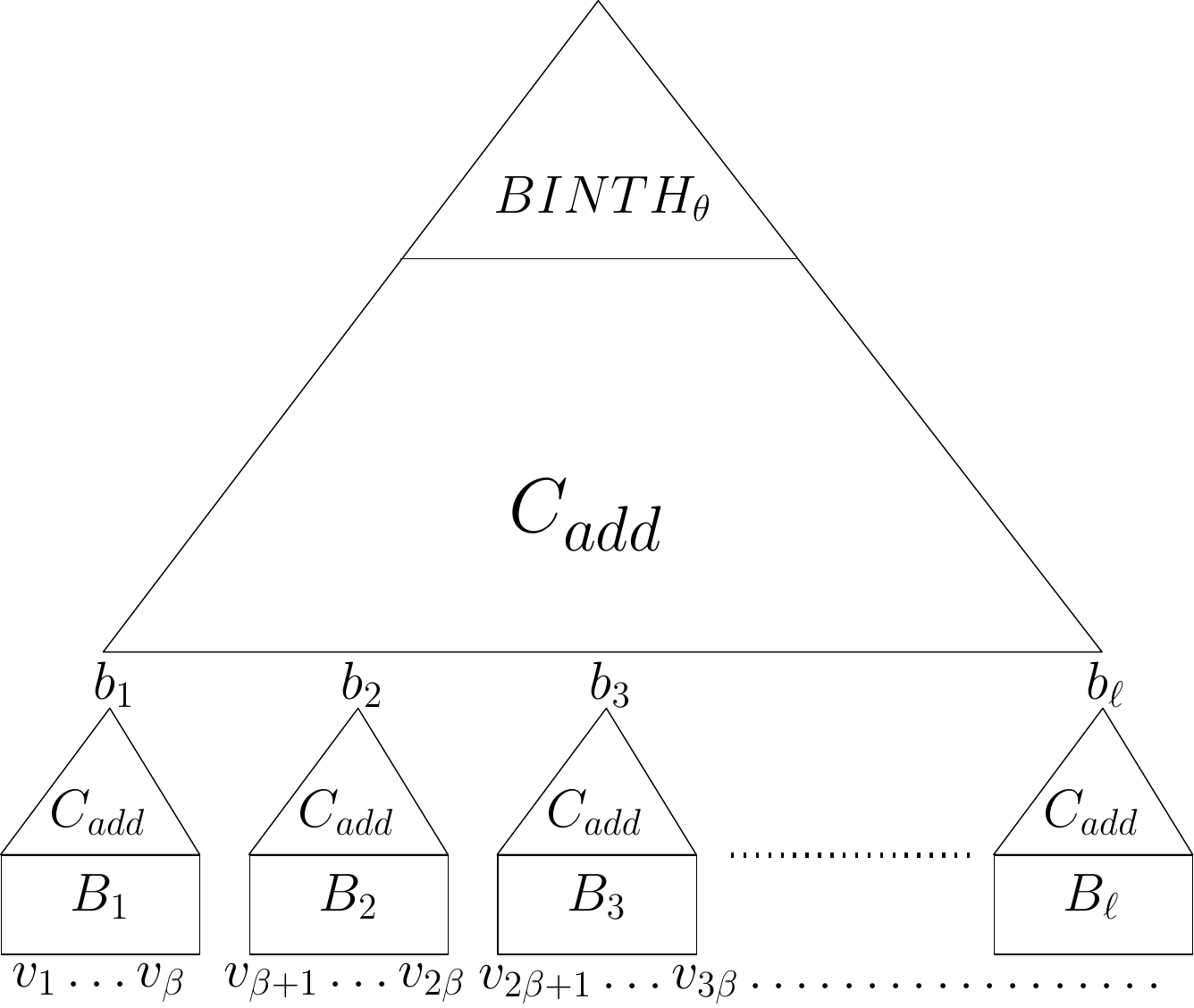}
  \caption{\label{fig:thgate}%
  Overview of replacement of $\gate{TH}_\theta$ gate.}
\end{figure}

  \paragraph*{Correctness of rewriting}
  To prove correctness, let us first observe that the transformation in C2
  does not change the functionality of~$C$ since we just explicitly simulate
  threshold gates with large fan-in by constructing a small Boolean circuit over
  the De Morgan basis in a straightforward way.
  We do this transformation in a block-wise fashion in order to save the number of
  additional variables we add in step C3.
  So let~$C$ be the circuit after its transformation in C2 and consider step C3.
  Let $x$ be a setting for the $n$ input gates~$I(C)$.
  We claim that $C(x)=1$ holds if and only if there is a setting for the
  $y$-variables such that $F(x,y)=1$.
  Indeed, if $C(x)=1$, we set~$y$ such that $y_a=C_a(x)$ holds for all $a\in A$.
  This setting for~$y$ then satisfies the constraint~``$y_v=C_{v,A}(x,y)$'' and thus
  $F_v(x,y)=1$.
  Moreover, $y_o=C(x)=1$ holds as well, and so the formula~$F$ constructed in C4 is
  satisfied by~$(x,y)$.
  For the reverse direction, let $(x,y)$ be such that $F(x,y)=1$.
  We claim that $C(x)=1$ holds as well.
  Indeed, by construction of~$F$, we know that~$y_o=1$ and $C_{v,A}(x)=y_v$
  holds for all $v\in A$.
  We can see by induction on the depth of~$v$ (starting at the bottom) that
  $C_{v}(x)=C_{v,A}(x,y)$ holds.
  In the base case, the only input gates of $C_{v,A}$ are the original $x$-input
  gates from~$I(C)$, and thus $y_v=C_{v,A}(x,y)=C_{v}(x)$ holds.
  In the inductive case, $C_{v,A}$ may depend on variables~$y_a$.
  However, for each such variable, we know by the induction hypothesis that
  $y_a=C_{a,A}(x,y)=C_{a}(x)$ holds, and thus we have $C_v(x)=C_{v,A}(x,y)$ by
  the definition of~$C_v$.
  In particular $C(x)=y_o=1$ holds and so~$x$ satisfies~$C$.
  This establishes the correctness of the reduction, except for proving that width $k = \beta^d$ is sufficient (in step C3b), which we will show next.

  \paragraph*{Bounding the width}
  In C3a, note that by the definition of $C_{v,A}$ and $R_{C}(A,v)$
  (see Section~\ref{sec:prel}), the value of gate $v$ on input $x$ is determined by the set of
  values of the gates $u \in (I(C) \cup A) \cap R_{C}(A,v)$.
  Therefore, in C3b we can ensure the value $y_v$ equals the value of gate $v$ on input $x$ by adding
  clauses on $y_v$ and the variables corresponding to the gates in $(I(C) \cup A) \cap R_{C}(A,v)$.
  We need to prove that the set $(I(C) \cup A) \cap R_{C}(A,v)$ has size less than
  $k=\beta^d$ so that this can be done in~$k$-CNF.
  For most gates~$v$ added to~$A$ this is clear because there are only two gates
  feeding into~$v$ after the replacement step C2d and both of these gates are
  in~$A$ as well.
  The only exceptions are the $b_1,\dots,b_\ell$-gates.
  Note that any $b_i$-gate $v$ is determined by the $B_i$-gates below it, so we can bound $|(I(C) \cup A) \cap R_{C}(A,v)| \le \sum_{u \in B_i}| (I(C) \cup A) \cap R_{C}(A,u)|$. Any gate $u \in B_i$ already existed in the original circuit. If $u$ has degree at least $\beta$ in the original circuit, then we ran step C2 on $u$ and thus $u$ belongs to $A$. Otherwise, $u$ has less than $\beta$ children, which already existed in the original circuit, and on which we recurse. It follows that if gate $u$ has depth $d_u$ in the original circuit, then it can be reached from less than $\beta^{d_u}$ nodes in $A$ without going through any other node in~$A$, i.e., $|(I(C) \cup A) \cap R_{C}(A,u)| < \beta^{d_u}$. Since $u$ is a descendant of $v$, we have $d_u < d$. In total, we obtain
  $|(I(C) \cup A) \cap R_{C}(A,v)| < |B_i| \cdot \beta^{d-1} < \beta^d$.
  It follows that the constraints in step C3b indeed consider at most
  $k=\beta^d$ variables and can thus be expressed in $k$-CNF.
  
  \paragraph*{Bounding the number of variables}
  It remains to set $\beta$ in such a way that $|A|$ is at most $\eps n$, which
  implies that~$F$ has at most $(1+\eps)n$ variables.
  Recall that the loop at C2 iterates over all gates $v$ with fan-in $d^-_C(v) \ge \beta$.
  For any such gate $v$, we add at most $50 b \ell$
  gates to~$A$ (see step C2d), where $b = \log \beta + 1$ and $\ell\le 1+d^-_C(v) / \beta \le 2 d^-_C(v) / \beta$ since $d^-_C(v) \ge \beta$. Hence, overall the number of gates ever added to~$A$ is at most
  \begin{align*}
    \sum_{\substack{v\in V(G)\\d^-_C(v)\ge\beta}}
    100 \, d^{-}_C(v)\cdot\frac{\log\beta + 1}{\beta}
    &\le 100\, cn \frac{\log\beta + 1}{\beta}
    \,.
  \end{align*}
  Thus we can set $\beta=\beta(c,\epsilon)\in\N$ as a function of~$c$ and
  $\epsilon$ such that the size of~$A$ is at most~$\epsilon n$.
  One can check that $\beta \le 2000 (c/\eps) \log(2c/\eps)$ suffices for $\eps \le 1 \le c$, where all logs are base 2. Thus, we get the claimed
  upper bound on~$k$.
  This concludes the proof of the lemma.
\end{proof}

Let us remark that in Lemma~\ref{lem:rewrite} we can also handle several other
gates, such as $\gate{MOD}_m$ gates, by replacing the $\gate{BINTH}_\theta$ circuit in the proof with a circuit that checks whether a given integer is a multiple of a given $m$. In fact, we can handle any symmetric gate $f(x_1,\ldots,x_d) = g(\sum_{i=1}^d x_i)$ where $g(s)$ can be expressed as a $o(d)$-size DeMorgan circuit when given $s \in \{0,\ldots,d\}$ in binary.

Using Lemma~\ref{lem:rewrite} for an $\eps$ with $(1+\eps)s_\infty < 1$ we obtain:

\begin{reptheorem}{theorem: TC0}
  \theoremTCzero
\end{reptheorem}

\subsection{Improving the Dependence in Depth to Sub-exponential}\label{sec: TC1}

In this section we improve the dependence of $k$ on $c,\eps$ and $d$ in Lemma~\ref{lem:rewrite}. As this dependence is exponential in $d$, it is natural to employ existing techniques for depth reduction of circuits, such as the following result due to Valiant~\cite{DBLP:conf/mfcs/Valiant77}. We use the following variant~\cite[Lemma 1.4]{Jukna:2012:BFC:2190632} (see also~\cite[Section 4.2]{TCS-033}).
\begin{lemma}\label{lem:valdep}
  In any directed graph with $m$ edges and depth $2^\delta$ (where $\delta$ is integral), it is possible to remove in polynomial time a set $R$ of $rm/\delta$ edges so that the depth of the resulting graph does not exceed $2^{\delta-r}$.
\end{lemma}

We use Lemma~\ref{lem:valdep} to improve the dependence of $k$ in Lemma~\ref{lem:rewrite}.

\begin{lemma}\label{lem:rewrite2}
  There is an algorithm that, given $\eps>0$ and a depth-$d$ threshold circuit with at most $cn$ wires, produces at most $2^{\eps n/2}$ $k$-CNF formulas $\varphi_1,\ldots,\varphi_z$ on $(1+\eps/2)n$ variables such that the circuit $C$ is satisfiable if and only if $\varphi_i$ is satisfiable for some $i \leq z$, and we have $k \leq (4000(c/\eps)\lg(4c / \eps))^{(2d)^{1-\eps/(2c)}}$.
\end{lemma}
\begin{proof}
  Let $C=(G=(V,E),\lambda)$ be the given circuit. Apply Lemma~\ref{lem:valdep}
  to $G$ with $\delta = \lceil \lg d \rceil$ and $r=\eps \delta/(2c)$. We obtain
  a set $R$ of size at most $\eps n/2$ such that $(V,E\setminus R)$ has depth at
  most $2^{\delta(1-\eps/(2c))}\leq (2d)^{1-\eps/(2c)}$. For every assignment $a
  \in \{0,1\}^R$ we create a circuit where we require that $C_x(v)=a_v$ for
  every $v \in R$, remove their outgoing wires, and update their incident gate
  accordingly (i.e. if $a_v=1$ and the incident gate is a $\gate{TH}_\theta$ it
  becomes a $\gate{TH}_{\theta-1}$ gate). The obtained circuit has depth at most $(2d)^{1-\eps/(2c)}$ and applying Lemma~\ref{lem:rewrite} gives the claimed result.	
\end{proof}

The improved dependence of $k$ in Lemma~\ref{lem:rewrite2} allows for the following consequence, yielding an exponential speedup for sparse threshold circuits of any depth $(\log n)^{1+o(1)}$.

\begin{reptheorem}{theorem: TC1}
  \theoremTCone
\end{reptheorem}
\begin{proof}
  Apply Lemma~\ref{lem:rewrite2} with $\eps/2$ to obtain $2^{\eps n/2}$
  $k$-CNF formulas
  on $(1+\eps/2)$ variables and
  with $k = \exp\paren*{O((\lg n)^{(1+\delta)(1-\eps/(2c))})}$.
  For sufficiently small $\delta = \delta(\eps,c) > 0$ we have $k = 2^{o(\log n)} = o(n/\log n)$ and thus the number of clauses $m$ is at most $(2n)^k = n^{o(n / \log n)} = 2^{o(n)}$. Therefore the assumed algorithm for CNF-SAT determines the satisfiability of the produced CNF formula in time $2^{\eps n/2} \cdot m^{O(1)} \cdot 2^{(1-\eps)(1+\eps/2)n}$, which is $O(2^{(1-\eps')n})$ for any $\eps' < \eps^2/2$.
\end{proof}

\paragraph*{Open Question.}
It is known that Lemma~\ref{lem:valdep} cannot be significantly improved
(see~\cite{SCHNITGER198289}). However, this does not stop us from using the
power of branching to get improvements. Specifically, when we try an assignments
of the truth-value on edges in $R$ in Lemma~\ref{lem:rewrite}, all gates that
are not connected to inputs are constant so these and their wires can already be
computed and removed from the circuit. A natural question is whether
this can be exploited more: Given a DAG $G=(V,E)$ of depth $2^{\delta}$ on $m$ edges and a real number $0 <
\alpha < 0$. For a set $R$ of edges, denote $l(R)$ as the length of the longest
path in $(V,E \setminus R)$ starting at a vertex $v$ which is a source in $G$. Give an upper bound on $\min\{ l(R): |R| \leq
\eps m \}$ better than $2^{(1-\eps)\delta}$ (which is implied by
Lemma~\ref{lem:valdep}).

  \paragraph*{Acknowledgments}
  Jesper Nederlof is supported by NWO Veni grant 639.021.438.
  We thank Ivan Mikhailin for a helpful comment on an earlier version of this manuscript.

	\bibliographystyle{plainurl}
	\bibliography{refs}

	\appendix
  \vfill
\section{Schematic Overview of our Results}\label{app}
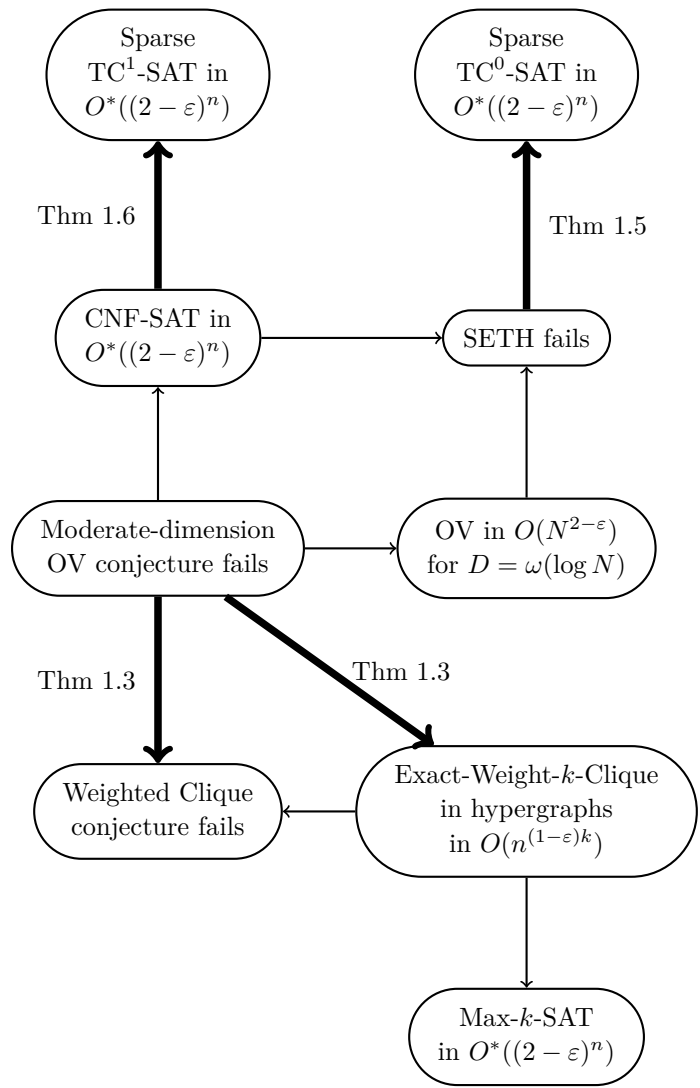
\begin{figure}[h]
  \centering
\tikzstyle{vecArrow} = [thick, decoration={markings,mark=at position
   1 with {\arrow[semithick]{open triangle 60}}},
   double distance=1.4pt, shorten >= 5.5pt,
   preaction = {decorate},
   postaction = {draw,line width=1.4pt, white,shorten >= 4.5pt}]
\tikzstyle{innerWhite} = [semithick, white,line width=1.4pt, shorten >= 4.5pt]

\noindent\begin{tikzpicture}[thick,every node/.style={scale=0.85},node distance=2.8cm,inner sep = 7pt,every node/.style=on grid]
\begin{scope}
  \node[draw,rounded rectangle,align=center] (a) {CNF-SAT in\\$O^*((2-\eps)^n)$};
  \node[draw,rounded rectangle,below of=a,align=center] (b)
    {Moderate-dimension\\OV conjecture fails};
  \node[draw,rounded rectangle,right of=b,align=center,xshift=2.1cm] (d) {OV in $O(N^{2-\eps})$ \\ for $D = \omega(\log N)$};
  \node[draw,rounded rectangle,above of=d,align=center] (e) {SETH fails};
  \node[draw,rounded rectangle,above of=e,align=center,yshift=0.7cm] (f) {Sparse \\ \TC0-SAT in \\ $O^*((2-\eps)^n)$};
  \node[draw,rounded rectangle,above of=a,align=center,yshift=0.7cm] (g) {Sparse \\ \TC1-SAT in \\ $O^*((2-\eps)^n)$};
  \node[draw,rounded rectangle,below of=b,align=center,yshift=-0.7cm] (h) {Weighted Clique \\ conjecture fails};
  \node[draw,rounded rectangle,below of=d,align=center,yshift=-0.7cm] (i) {Exact-Weight-$k$-Clique \\ in hypergraphs \\ in $O(n^{(1-\eps)k})$};
  \node[draw,rounded rectangle,below of=i,align=center,yshift=-0.2cm] (j) {Max-$k$-SAT \\ in $O^*((2-\eps)^n)$};
  
  \draw[->] (b) to (a);
  \draw[->] (a) to (e);
  \draw[->] (d) to (e);
  \draw[->] (b) to (d);
  \draw[->,line width=1mm] (e) -- (f) node[midway,right] {Thm~\ref{theorem: TC0}};
  \draw[->,line width=1mm] (a) -- (g) node[midway,left] {Thm~\ref{theorem: TC1}};
  \draw[->] (i) to (j);
  \draw[->] (i) to (h);
  \draw[->,line width=1mm] (b) -- (h) node[midway,left] {Thm~\ref{theorem: OV main}};
  \draw[->,line width=1mm] (b) -- (i) node[midway,right] {Thm~\ref{theorem: OV main}};
\end{scope}
\end{tikzpicture}
  \caption{\label{fig2}%
	An overview of relevant implications. New implications presented in this paper are displayed with bold arcs and labeled with the theorem number.}
\end{figure}

\end{document}